\documentclass[journal,twoside,web,hidelinks]{ieeecolor}
\usepackage{generic}
\usepackage{orcidlink}

\usepackage{amsmath,mathtools,amssymb}

\usepackage{algorithm,algpseudocode,algorithmicx}
\usepackage{cite,cleveref}

\usepackage{url}
\usepackage{textcomp}
\usepackage[normalem]{ulem}
\usepackage{graphicx}

\graphicspath{{figures/}}

\newtheorem{theorem}{Theorem}
\newtheorem{cor}{Corollary}

\newtheorem{prop}{Proposition}

\newtheorem{defn}{Definition}

\newcommand{\bs}{\boldsymbol}
\newcommand{\bb}{\mathbb}
\newcommand{\mcal}{\mathcal}

\newcommand{\eye}{\bs{I}}
\newcommand{\zero}{\bs{0}}
\newcommand{\one}{\bs{1}}

\newcommand{\lb}{\left(}
\newcommand{\rb}{\right)}
\newcommand{\ls}{\left[}
\newcommand{\rs}{\right]}
\newcommand{\lc}{\left\{}
\newcommand{\rc}{\right\}}

\newcommand{\lv}{\left\vert}
\newcommand{\rv}{\right\vert}
\newcommand{\lV}{\left\Vert}
\newcommand{\rV}{\right\Vert}
\newcommand{\LRV}[1]{{\left\vert\kern-0.25ex\left\vert\kern-0.25ex\left\vert #1 \right\vert\kern-0.25ex\right\vert\kern-0.25ex\right\vert}}
\newcommand{\matV}[1]{{\left\vert\kern-0.25ex\left\vert\kern-0.25ex\left\vert #1 \right\vert\kern-0.25ex\right\vert\kern-0.25ex\right\vert}}

\newcommand{\T}{\mathsf{T}}
\newcommand{\trace}[1]{\mathsf{Tr}\lc#1\rc}

\newcommand{\rank}[1]{\mathsf{Rank}\lc#1\rc}

\newcommand{\matA}{\bs{A}}
\newcommand{\matB}{\bs{B}}
\newcommand{\matC}{\bs{C}}

\newcommand{\matI}{\bs{I}}

\newcommand{\matK}{\bs{K}}
\newcommand{\matL}{\bs{L}}
\newcommand{\matM}{\bs{M}}

\newcommand{\matO}{\bs{O}}
\newcommand{\matP}{\bs{P}}
\newcommand{\matQ}{\bs{Q}}
\newcommand{\matR}{\bs{R}}
\newcommand{\matS}{\bs{S}}

\newcommand{\matW}{\bs{W}}

\newcommand{\bbE}{\bb{E}}

\newcommand{\bbR}{\bb{R}}

\newcommand{\calA}{\mcal{A}}
\newcommand{\calB}{\mcal{B}}

\newcommand{\calG}{\mcal{G}}

\newcommand{\calI}{\mcal{I}}

\newcommand{\calN}{\mcal{N}}
\newcommand{\calO}{\mcal{O}}

\newcommand{\calS}{\mcal{S}}
\newcommand{\calT}{\mcal{T}}
\newcommand{\calU}{\mcal{U}}
\newcommand{\calV}{\mcal{V}}

\newcommand{\vecd}{\bs{d}}

\newcommand{\vecu}{\bs{u}}
\newcommand{\vecv}{\bs{v}}
\newcommand{\vecw}{\bs{w}}
\newcommand{\vecx}{\bs{x}}
\newcommand{\vecy}{\bs{y}}
\newcommand{\vecz}{\bs{z}}

\newcommand{\vecmu}{\bs{\mu}}

\newcommand{\matSigma}{\bs{\Sigma}}

\newcommand{\eps}{\epsilon}

\newcommand{\norm}[2]{\left\lVert#1\right\rVert_{#2}}

\newcommand{\set}[1]{\left\{#1\right\}}

\newcommand*{\qedb}{\null\nobreak\hfill\ensuremath{\blacksquare}}
\Crefname{figure}{Fig.}{Figs.}
\allowdisplaybreaks

\def\BibTeX{{\rm B\kern-.05em{\sc i\kern-.025em b}\kern-.08em
    T\kern-.1667em\lower.7ex\hbox{E}\kern-.125emX}}
\begin{document}
\title{Discrete-Time Linear Dynamical System Control Using Sparse Inputs With Time-Varying Support}

\author{Krishna Praveen V. S. Kondapi$^{1\orcidlink{0009-0003-4936-1288}}$, Chandrasekhar Sriram$^{2\orcidlink{0000-0002-5503-4334}}$, Geethu Joseph$^{3\orcidlink{0000-0002-5289-5403}}$ and Chandra R. Murthy$^{1\orcidlink{0000-0003-4901-9434}}$ 
\thanks{$^{1}$K. Kondapi and C. R. Murthy are with the Dept. of ECE at the Indian Institute of Science (IISc), Bangalore 560012, India. Emails:
        {\tt\small \{praveenkvsk, cmurthy\}@iisc.ac.in}.}%
        \thanks{$^{2}$C. Sriram is with Texas Instruments India Pvt. Ltd. He was at the Dept. of ECE, IISc, during the course of this work. Email: {\tt\small chandrasekhars@alum.iisc.ac.in}}
\thanks{$^{3}$G. Joseph is with the Faculty of Electrical Engineering, Mathematics, and Computer Science at the Delft University of Technology, Delft 2628 CD,  Netherlands. Email:
        {\tt\small G.Joseph@tudelft.nl}.}
\thanks{A part of this work was presented in \cite{Kondapi_icc_2024}.}
}

\maketitle

\begin{abstract}
In networked control systems, communication resource constraints often necessitate the use of \emph{sparse} control input vectors. A prototypical problem is how to ensure controllability of a linear dynamical system when only a limited number of actuators (inputs) can be active at each time step. In this work, we first present an algorithm for determining the \emph{sparse actuator schedule}, i.e., the sequence of supports of the input vectors that ensures controllability. Next, we extend the algorithm to minimize the average control energy by simultaneously minimizing the trace of the controllability Gramian, under the sparsity constraints. We derive theoretical guarantees for both algorithms: the first algorithm ensures controllability with a minimal number of control inputs at a given sparsity level; for the second algorithm, we derive an upper bound on the average control energy under the resulting actuator schedule. Finally, we develop a novel sparse controller based on Kalman filtering and sparse signal recovery that drives the system to a desired state in the presence of process and measurement noise. We also derive an upper bound on the steady-state MSE attained by the algorithm. We corroborate our theoretical results using numerical simulations and illustrate that sparse control achieves a control performance comparable to the fully actuated systems. 
\end{abstract}


\section{Introduction}
\label{sec:intro}
The field of networked control systems, in which the controllers, sensors, and actuators communicate over a band-limited network to achieve a given control objective, is a growing area of research \cite{control_comm_constraint, NAGAHARA201184, jadbabaie2019Minreach, Ge_2020_DynamicEvtControl,  Tzoumas2022RASeqOpt, Tzoumas_2021_LQGControlSensing, Lu_2023_ControlCommSch}. Due to data-rate constraints imposed by the network, these systems demand bandwidth-efficient control inputs. One way to reduce the bandwidth utilization is to use only a few of the available actuators at each time instant, i.e.,  the number of nonzero entries of the control input is small compared to its length, leading to a \emph{sparse control input}~\cite{Li_2016_CSinSwitchedSys, NAGAHARA201184, Tzoumas2022RASeqOpt}. Sparse vectors admit compact representation in a suitable basis~\cite{foucart2013math}, and thus, save bandwidth. Sparse control also finds applications in resource-constrained systems such as environmental control systems and network opinion dynamics \cite{joseph2021controllability, WENDT201937}. Motivated by this, we consider the problem of designing sparse control inputs to drive the system to any desired state in a discrete-time linear dynamical system (LDS). 

The problem of finding an actuator schedule, i.e., the sequence of supports (indices of nonzero entries) of sparse inputs to ensure controllability of the discrete-time LDS, has attracted significant research interest in the last decade~\cite{gjoseph2020controllability, Kondapi_icc_2024, Ballotta_2024_PtwiseSch, Zhao_2016_NodeSchControl, Olshevesky_2014_MinControl, jadbabaie2019Minreach,  jadbabaie2018deterministic, SIAMI_2020_Separation}. Some previous studies considered the support of the inputs to be fixed across time~\cite{Olshevesky_2014_MinControl, jadbabaie2019Minreach, Fotiadis2024DataDrivenAlloc}. It was shown that finding the minimum number of actuators needed to control the system is NP-Hard~\cite{Olshevesky_2014_MinControl, jadbabaie2019Minreach}. 
Time-varying supports in the control inputs (a time-varying schedule) can potentially ensure the controllability of systems that might be uncontrollable when using inputs with fixed support (time-invariant schedule). 
Scenarios where a time-varying schedule is better than a time-invariant schedule were discussed in~\cite{Nozari_2017_Tinv}.
One variant of this problem only limits (or minimizes) the \emph{average} number of actuators utilized across all time steps~\cite{Nagahara_2016_MaxHandsOff, jadbabaie2018deterministic, SIAMI_2020_Separation}. In~\cite{jadbabaie2018deterministic, SIAMI_2020_Separation} algorithms to obtain an actuator schedule such that the resulting Gramian matrix approximates the Gramian matrix of the fully actuated system are developed. 
Further,  another line of work in~\cite{nishida2024sparseLQR, Nishida2024BalancingQuadratic} constrained the control inputs to be nonzero at only a few time instants. Since the nonzero inputs need not be sparse, transmitting these control inputs in a networked control system requires a large peak data rate even though the average rate is low. 
In contrast, this paper considers time-varying schedules where the cardinality of the control input support (hence the required network bandwidth) remains the same across time~\cite{gjoseph2020controllability, Kondapi_icc_2024, Ballotta_2024_PtwiseSch, Zhao_2016_NodeSchControl}.

The notion of controllability of a discrete-time LDS using sparse inputs (called sparse controllability) was defined, and the conditions for sparse controllability were derived in \cite{gjoseph2020controllability}. This work proved the \emph{existence} of an actuator schedule that ensures controllability.
The authors in~\cite{Zhao_2016_NodeSchControl} developed an algorithm that outputs a time-varying sparse actuator schedule that minimizes the worst-case control energy with a symmetric transfer matrix. However, no explicit controllability guarantees were derived. Also, the designed algorithm has combinatorial complexity. In~\cite{Kondapi_icc_2024}, the authors developed a greedy algorithm to minimize the inverse of the controllability Gramian. However, this algorithm is not always guaranteed to ensure controllability~\cite{Ballotta_2024_PtwiseSch}. 
This limitation is overcome to some extent in~\cite{Ballotta_2024_PtwiseSch}, where the authors present a scheduling algorithm that selects a few necessary actuators to ensure controllability. 
In the sequel, we provide an example where their algorithm also fails to ensure controllability. We develop another low-complexity scheduling algorithm with a theoretical controllability guarantee under mild conditions (See \Cref{thm:FullB_Ctrl_Guarantee}).

Recently, actuator selection with linear-quadratic regulator and linear-quadratic Gaussian objectives was considered in~\cite{Chamon_2019_MatroidOpt, Chamon2022ApproxSuperMatroid, Jiao_2023_ActAchCvxRlx, Lintao_2025_OnlineActSct}. In~\cite{Chamon_2019_MatroidOpt, Chamon2022ApproxSuperMatroid, Jiao_2023_ActAchCvxRlx}, a dynamic programming approach was used to formulate an actuator scheduling problem, which was solved using greedy and convex relaxation approaches. These studies considered a time-varying schedule, except for~\cite{Lintao_2025_OnlineActSct}. In~\cite{Lintao_2025_OnlineActSct}, an actuator selection problem was formulated as a multiarmed bandit problem when the dynamics of the LDS are unknown. 
All of these papers assumed the availability of perfect state information at the controller at each time instant, which may not be feasible in practice. Therefore, we also consider an LDS with additive noise in the system state evolution as well as measurement equations. In~\cite{Vafaee2023LargeSensorNwk}, the authors considered a sensor scheduling problem in an LDS with process and measurement noise. They proposed a greedy algorithm that minimizes the trace of the covariance matrix of the Kalman filter at each time instant. In contrast, our work develops an algorithm to minimize the tracking error between the system state and target state using inputs obtained from a sparse recovery algorithm at each time instant.


In short, we first consider a noiseless LDS and formulate the actuator scheduling problem as minimizing the trace of the inverse of the controllability Gramian, a widely used metric for quantifying average control energy~\cite {jadbabaie2018deterministic}. Then, we consider an LDS with measurement and process noise and formulate an optimization problem to minimize the tracking error at each time instant. We solve these problems subject to a constraint limiting the active actuators to at most $s$ per time step. The specific contributions of this paper are as follows:
\begin{itemize}
    \item We present an algorithm that returns an actuator schedule that is guaranteed to ensure controllability when the input matrix has full row rank. Our algorithm outputs a schedule that allows one to reach any desired state:
    a) using the minimum value of $s$, and b) using the minimum number of control inputs (See \Cref{sec:FullRankB_CtrlGuarantee}.)
	\item We study controllability under energy constraints by finding an actuator schedule that minimizes a cost function related to average control energy. We show that our cost function exhibits $\alpha-$supermodularity, and the sparsity constraint is a matroid. This leads to a greedy algorithm with provable guarantees for selecting the sparse actuator schedule. (See \Cref{sec:greedy_sch}.)
    \item We develop a Kalman filter based algorithm to design sparse inputs at each time instant to control the LDS with process and measurement noise to reach a target state $\vecx_f$. Additionally, we derive an upper bound on its steady-state mean squared error (MSE). (See \Cref{sec:Nsy_Sys}.)
\end{itemize}
Overall, we systematically investigate the controllability of an LDS using sparse inputs with time-varying support. 
To the best of our knowledge, ours is the first actuator scheduling algorithm that ensures controllability with fixed cardinality support. Also, for a noisy system, we present an algorithm that generates $s$-sparse control inputs using output measurements.

\textbf{Notation:} 
Boldface capital letters, boldface small letters, and calligraphic letters denote matrices, vectors, and sets. 
The $\ell_0$-norm, which counts the number of nonzero elements, is denoted by $\Vert \cdot \Vert_0$,  and the set of indices of the nonzero entries is called the support. Also, $\Vert \cdot \Vert$ denotes the induced $\ell_2$-norm for matrices and Euclidean norm for vectors. The identity matrix is denoted by $\matI$. We use $\matA_{\calS}$ to denote the submatrix of $\matA$ formed by the columns indexed by the set $\calS$. Similarly, $\matA_i$ denotes the $i$th column of the matrix $\matA$. The column space (range) and null space (kernel) of a matrix $\matA$ are denoted as $\mcal{CS}(\matA)$ and $\mcal{NS}(\matA)$, respectively. The symbol $\lceil \cdot \rceil$ denotes the ceiling function. We use $2^{\calV}$ to denote the power set of $\calV$. The operator $\mathbb{E}$ represents expectation and $\mathcal{N}(\vecmu, \matC)$ denotes the Gaussian distribution with mean $\vecmu$ and covariance matrix $\matC$. Also, $\one$ denotes the all-one vector, and $\zero$ denotes the all-zero vector or matrix. Finally, $\matA \succ \zero$ and $\matA \succeq \zero$ represent that $\matA$ is positive definite and positive semi-definite, respectively.

\section{Actuator Scheduling Problem}
\label{sec:sys_model}
We consider a discrete-time LDS $(\matA,\matB)$  whose dynamics are governed by 
\begin{equation}
\vecx (k+1) = \matA \vecx (k) + \matB \vecu(k), \label{eq:state}
\end{equation}
where $k = 0,1,\ldots$ is the integer time index, $\vecx (k) \in \mathbb{R}^n$ and $\vecu(k) \in \mathbb{R}^m$ are the system state and the input vectors, respectively, at time $k$. The matrices $\matA \in \mathbb{R}^{n \times n}$ and $\matB \in \mathbb{R}^{n \times m}$ are the system transfer matrix and the input matrix, respectively. In~\cite{Olshevesky_2014_MinControl, Nozari_2017_Tinv, Bof_2017_NwkCentrality}, the input matrix is taken as $\matB = \matI$. In the sequel, we assume $\rank{\matB}=n$; no other conditions on $\matB$ are assumed. We consider the input $\vecu(k)$ to have at most $s$ nonzero entries ($s$-sparse inputs), i.e., $\Vert\vecu(k)\Vert_0\leq s$, for $k = \{0, 1, \ldots, K-1\}$. 

A classical control design problem is to estimate the control inputs $\{\vecu(k)\}_{k=0}^{K-1}$ that drive the system state to a desired state $\vecx_f \in \bbR^{n}$. 
If the support of $\vecu(k)$ is known, we can easily find the sparse inputs using conventional methods. However, sparsity constraints make finding such an actuator schedule (support of the sparse inputs) challenging. Further, it is known that a sparse actuator schedule that ensures controllability exists if and only if the system is \emph{$s$-sparse controllable}~\cite{gjoseph2020controllability}. Yet, \cite{gjoseph2020controllability} does not provide a method to determine the actuator schedule for $s$-sparse controllable systems. In this work, we address this gap by designing a sparse actuator schedule that can be used to drive the system to any given desired state in $K$ time steps ($\vecx(K)=\vecx_f$). 

The $s$-sparse controllability is equivalent to the following two conditions~\cite{gjoseph2020controllability}: the system is controllable, and $s\geq n-\rank{\matA}$.
So, in the sequel, to ensure that the problem is feasible, we assume that these two conditions hold. 
Further, for an $s$-sparse controllable LDS, the minimum number of time steps $K$ required to ensure controllability is bounded as
\begin{equation} \label{eq:sparse_cntrl_bounds}
\frac{n}{s} \leq K \leq \min\lb q \left\lceil \frac{\rank{\matB}}{s} \right\rceil, n-s+1 \rb.
\end{equation}
where $q$ is the degree of the minimal polynomial of $\matA$. 
Hence, with $\matR$ denoting the controllability matrix, there exists a finite $K$ and an actuator schedule $\calS \triangleq  (\calS_0,\calS_1,\ldots,\calS_{K-1})$ with $\vert \calS_k \vert \leq s$ such that, 
\begin{equation} \label{eq:control_con}
    \matR_\calS \triangleq \begin{bmatrix} 
    \matA^{K-1} \matB_{\calS_0} & \matA^{K-2} \matB_{\calS_1} & \ldots & \matB_{\calS_{K-1}},
    \end{bmatrix}
\end{equation}
satisfies $\rank{\matR_\calS}=n$. Once $\calS$ is determined, the system can be driven from any initial state to a desired final state by applying inputs $\vecu(k)$ with support $\calS_k$, for $k = 0, 1, \ldots, K-1$. Thus, we aim to find $\calS \in \Phi$ such that $\rank{\matR_\calS}=n$, where 
\begin{equation}\label{eq:feasible_set}
    \Phi \triangleq \lc(\calS_0,\ldots,\calS_{K-1}):\calS_k\subseteq \{1,\ldots,m\}, \vert \calS_k \vert \leq s \, \forall k\rc.
\end{equation}
We observe that any actuator schedule of length $K$ can be represented using a subset of  
\begin{equation}\label{eq:V_defn}
    \calV = \{ (k,j) \vert k=0,1,\ldots,K-1, \;j=1,2,\ldots,m \}.
\end{equation}
Here, $k$ and $j$ represent the time and actuator index, respectively. It is clear that there is a natural bijection between any $\calT\in 2^\calV$ and the corresponding actuator schedule:
\begin{equation*}
    \calS(\calT) = (\calS_0, \calS_1, \ldots,\calS_{K-1})\;\;\text{with}\;\; \calS_k=\{j: (k,j)\in\calT\}.
\end{equation*}
Also, including the tuple $(k,j)$ in the schedule $\mathcal{S}$ corresponds to including the $j$th actuator in the $k$th time step, i.e., the column $\matA^{K-k-1}\matB_{j}$ is included in $\matR_\calS$.
Clearly, brute-force search for $\calS \in \Phi$ is infeasible since $|\Phi| = \left(\sum_{l=0}^s \binom{m}{l}\right)^{K}$ can be prohibitively large.  In the next section, we present an algorithm that returns an actuator schedule that guarantees controllability subject to \eqref{eq:feasible_set}.

\section{Actuator Scheduling Algorithm to Guarantee Controllability} \label{sec:FullRankB_CtrlGuarantee}
We consider an actuator scheduling algorithm that starts with an empty set $\calS$ and incrementally adds entries to it until $\rank{\matR_\calS}=n$. Since $\rank{\matB} = n$, we have $\mcal{CS}({\matA^i\matB}) \subseteq \mcal{CS}({\matA^j\matB})$ for $i \geq j \geq 0$. This subspace structure allows for the addition of at most $s$ columns from $\matA^{K-1}\matB$ first, then adding at most $s$ columns from $\matA^{K-2}\matB$, and so on, increasing the rank till $\rank{\matR_\calS}=n$. The main result of this section is that the schedule returned by such a procedure ensures $s$-sparse controllability. 
First, for each $0 \leq i \leq K-1$, consider a set $\calT^{(i)}$ containing column indices (expressed as tuples of the form defined above) of $\begin{bmatrix}
    \matA^{K-1}\matB & \matA^{K-2}\matB & \ldots & \matA^{i}\matB
\end{bmatrix}$ such that the columns indexed by $\calT^{(i+1)}$ are linearly independent (LI) and $\calT^{(K)} = \emptyset$. Let $\calI(i)$ denote a set containing $l(i)$ indices corresponding to columns of $\matA^{i}\matB$ such that $\calT^{(i+1)}\cup\calI(i)$ contain indices of LI columns in $\begin{bmatrix}
    \matA^{K-1}\matB & \matA^{K-2}\matB & \ldots & \matA^{i}\matB 
\end{bmatrix}$. We note that $\calI(i)$ contains $l(i)$ distinct tuples: $\calI(i) = \lc (K-i-1,j_l) \in \calV \rc_{l=1}^{l(i)}$, $j_l \in \set{1,2,\ldots,m}$ and satisfies 
\begin{equation} \label{eq:IndependentSet}
    \mcal{NS}\lb\matR_{\calS \lb \calI(i) \cup \calT^{(i+1)}\rb}\rb = \emptyset. 
\end{equation} 
The following proposition specifies the value of $l(i)$:
\begin{prop} \label{prop:MaxLIColumns}
    If $\rank{\matB}=n$, then there exist $\rank{\matA^i\matB}-\lv \calT^{(i+1)} \rv \geq 0$ linearly independent columns in $\matA^{i}\matB$ that do not belong to $\mcal{CS}(\matR_{\calS(\calT^{(i+1)})})$.
\end{prop}

We get the above result by noting that $\mcal{CS}\{\matR_{\calS(\calT^{(i+1)})}\} \subseteq \mcal{CS}\{\matA^i\matB\}$ and $\rank{\matR_{\calS(\calT^{(i+1)})}} = \lv \calT^{{i+1}} \rv$.

The above proposition ensures that we can set $l(i) = \min \{ s, \rank{\matA^i\matB}-\lv \calT^{(i+1)} \rv\}$ at each $i$ and ensure that \eqref{eq:IndependentSet} is satisfied. Then, $\calI(i)$ contains $l(i)$ indices of LI columns from the matrix $\matA^{i}\matB$ satisfying \eqref{eq:IndependentSet}, so that adding the set $\calI(i)$ to $\calT^{(i+1)}$ is guaranteed to increase the rank by $l(i)$. One way to compute the set $\calI(i)$ is by using the row echelon form of $\begin{bmatrix}
    \matR_{\calS(\calT^{(i+1)})} \! & \! \matA^i\matB
\end{bmatrix}$. Hence, finding $\calI(i)$ is computationally inexpensive. We present the overall procedure in \Cref{alg:LI_Schedule_FullB}, which returns an actuator schedule that guarantees controllability (see \Cref{thm:FullB_Ctrl_Guarantee}.) It is easy to show that the computational complexity of \Cref{alg:LI_Schedule_FullB} is $\calO(Kmn^2)$.

\begin{algorithm}[t]
    \caption{Controllable Time-Varying Actuator Schedule}
    \begin{algorithmic}[1]
        \Require System matrices $\matA, \matB$; time steps $K$; sparsity level~$s$
        \Statex \textbf{Initialization: } $\calT^{(K)} = \emptyset$
        \For{$i=K-1,K-2,\ldots,0$}
            \State $l(i) \leftarrow \min \lc s, \rank{\matA^i\matB} - \lv \calT^{(i+1)} \rv \rc$ 
            \State Find a set $\calI(i)$ using \eqref{eq:IndependentSet}
            \State $\calT^{(i)} \leftarrow \calI(i) \cup \calT^{(i+1)}$
        \EndFor
        \Ensure Schedule $\calG_0 = \calS(\calT^{(0)})$ 
    \end{algorithmic} \label{alg:LI_Schedule_FullB}
\end{algorithm} 

We note that \Cref{alg:LI_Schedule_FullB} adds the maximum possible number of LI columns from $\matA^i\matB$ while satisfying the sparsity constraint at each iteration. The value of $l(i)$ will become zero once $n$ linearly independent columns have been added to $\matR_{\calS(\calT^{(i+1)})}$. This scenario can occur well before the iteration index $i$ reaches zero, depending on the sparsity level $s$.

Finally, we show that \Cref{alg:LI_Schedule_FullB} returns a full rank schedule, ensuring controllability under the sparsity constraint~\eqref{eq:feasible_set}.
\begin{theorem} \label{thm:FullB_Ctrl_Guarantee}
    If $\rank{\matB}=n$, then for any $s\geq \max\lc 1,n-\rank{\matA} \rc$ and $K\geq \lceil \frac{n}{s} \rceil$,  \Cref{alg:LI_Schedule_FullB} finds an $s$-sparse actuator schedule $\calG_0$ such that $\rank{\matR_{\calG_0}}=n$. 
\end{theorem}
\begin{proof}
    See \Cref{app:FullB_Ctrl_Guarantee}.
\end{proof}

\Cref{alg:LI_Schedule_FullB} guarantees controllability with minimum sparsity level, $s=n-\rank{\matA}$, and also with minimum number of control inputs $(K=\lceil \frac{n}{s} \rceil)$ for a given $s$. Hence, it is optimal in both the sparsity level $(s)$ and the number of control inputs $(K)$ in the sense of providing a controllability guarantee.

\Cref{alg:LI_Schedule_FullB} returns a sparse actuator schedule corresponding to exactly $n$ columns of the controllability matrix  whereas the sparsity constraint \eqref{eq:feasible_set} allows the actuator schedule to contain $Ks\geq n$ pairs. Thus, there is room to choose additional columns in the controllability matrix while satisfying \eqref{eq:feasible_set} to potentially optimize the control energy or any other control objective. Specifically, in this work, we consider minimizing $\trace{\matW_{\calS}^{-1}}$, which represents the average energy to drive the LDS from $\vecx(0)=\zero$ to a uniformly random point on the unit sphere, and is a popular measure of the ``difficulty of controllability''~\cite{Olshevsky_2018_NonSupermodular}. Here, $\matW_\calS$ is the controllability Gramian, 
\begin{equation}\label{eq:W_mat_Defn}
	\matW_{\calS} = \matR_\calS\matR_\calS^{\mathsf{T}} = \sum_{k=1}^{K} \matA^{k-1} \matB_{\calS_{K-k}} \matB_{\calS_{K-k}}^{\T} (\matA^{k-1})^{\T}.
\end{equation}
We now present a greedy algorithm to minimize~$\trace{ \matW_{\calS}^{-1}}$.

\section{Greedy Scheduling Algorithm to Minimize Control Energy} \label{sec:greedy_sch}
First, we note that the metric $\trace{ \matW_{\calS}^{-1}}$ is not well-defined if $\matW_\calS$ is rank-deficient. So, we use the $\calG_0$ returned by \Cref{alg:LI_Schedule_FullB} as an initial schedule, which ensures that the metric is well-defined as the inverse exists. Thus, we consider the following optimization problem,
\begin{equation}
 \label{eq:opt_energy_epsilon_mod}
\underset{\calT \in  2^{\calV \setminus \calT_0}}\min \; \trace { (\matW_{\calS(\calT\cup\calT_0)})^{-1} }
\;\text{s.t.} \;\calS(\calT\cup\calT_0)\in\Phi,
\end{equation}
where $\Phi$ is defined in \eqref{eq:feasible_set}, $\calT_0$ is the set satisfying $\calS(\calT_0)=\calG_0$ and $\matW_\calS$ is given by \eqref{eq:W_mat_Defn}. The above optimization is a combinatorial problem. We will show that the objective function is $\alpha$-supermodular and the constraint $\Phi$ is a matroid, motivating a greedy approach to approximately solve~\eqref{eq:opt_energy_epsilon_mod}.

The greedy algorithm starts with the set of indices $\mathcal{T}^{(1)} = \calT_0$ and finds the element from $\calV\setminus\calT_0$ that minimizes the cost function and satisfies \eqref{eq:feasible_set} when added to $\calT^{(1)}$. Specifically, in the $r$th iteration of the algorithm, let $\calT^{(r)}$ be the set of indices collected up to the previous iteration. Then, we find
\begin{equation}\label{eq:greedy_update}
    (k^*,j^*) = \underset{(k,j)\in  \calV^{(r)}}{\arg\min} \trace { \matW_{\calS(\calT^{(r)}\cup\{(k,j)\})}^{-1} },
\end{equation}
where $\calV^{(r)} \subseteq \calV\setminus\calT^{(r)}$ is obtained by removing all infeasible index pairs, i.e., 
\begin{equation} \label{eq:search_space}
    \calV^{(r)} = \left\{ (k,j) \in \calV \setminus \calT^{(r)}: \calS(\calT^{(r)} \cup (k,j)) \in \Phi \right\}.
\end{equation}
We note that $\calT_0 \subseteq \calT^{(r)}$ for all $r\geq1$. Therefore, the set $\calT^{(Ks+1-n)}$ is an approximate solution to problem \eqref{eq:opt_energy_epsilon_mod}. The overall procedure is outlined in \Cref{alg:greedy_algo_epsilon}, which we call the \emph{RB$n$-greedy actuator scheduling algorithm}. We note that $\calT^{(i)}$ from \Cref{sec:FullRankB_CtrlGuarantee} is different from $\calT^{(r)}$, as these are specific to the respective algorithms. The computational complexity of \Cref{alg:greedy_algo_epsilon} is $\calO(msK^2n^2)$.

\begin{algorithm}[t]
	\caption{RB$n$-greedy actuator scheduling} 
	\begin{algorithmic} [1]
		\Require System matrices $\matA,\matB$; sparsity level~$s$.
        \Statex \textbf{Initialization: } $r=1$, $\calT^{(r)} = \calT_0$, $\calV^{(r)}$ from \eqref{eq:search_space} 
    	\While{$\calV^{(r)} \neq \emptyset$}
            \State Find $(k^*,j^*)$ using \eqref{eq:greedy_update}
        	\State $\calT^{(r+1)}=\calT^{(r)}\cup\{(k^*,j^*)\}$
            \State $\calG =(\calG_0,\calG_1,\ldots,\calG_{K-1})= \calS(\calT^{(r+1)})$
            \If{ $\vert \calG_{k^*}\vert= s$}
        	\State $\calV^{(r+1)}= \calV^{(r)}\setminus \{(k^*,j), \; j=1,2,\ldots,m\}$
            \Else
            \State $\calV^{(r+1)} = \calV^{(r)}\setminus \{(k^*,j^*)\}$
        		\EndIf
        	\State $r \leftarrow r+1$
    	\EndWhile
		\Ensure Actuator schedule $\calG$
	\end{algorithmic}
	\label{alg:greedy_algo_epsilon}
\end{algorithm}

We have $\trace{\matW_\calG^{-1}} \leq \trace{\matW_{\calG_0}^{-1}}$ because \Cref{alg:greedy_algo_epsilon} adds a positive semi definite matrix of rank 1 to $\matW_{\calS(\calT^{(r)})}$. Hence, $\trace{\matW_{\calS(\calT^{(r+1)})}^{-1}}$ does not increase when $(k^*,j^*)$ is added to $\calT^{(r)}$. This shows that \Cref{alg:greedy_algo_epsilon} reduces the average control energy without violating the sparsity constraint.

Next, to analyze the performance of \Cref{alg:greedy_algo_epsilon}, we introduce the notions of supermodularity and matroid constraints: 
\begin{defn}[Supermodularity]\label{defn:submodular}
The set function $f : \calU \to \mathbb{R}$ is said to be $\alpha$-submodular if $\alpha \in \mathbb R_{+}$ is the largest number for which the following holds $\forall \calA \subseteq \calB \subseteq \calU$ and $\forall e \in \calU \setminus \calB$:
\begin{equation}
f(\calA \cup \{e\} ) - f(\calA) \geq {\alpha [f(\calB \cup \{e\} ) - f(\calB)]}.
\end{equation}
Also, the function $-f(\cdot)$ is said to be $\alpha$-supermodular. 
\end{defn}
\begin{defn}[Matroid]\label{defn:matroid}
A matroid is a pair $(\calU, \calI)$ where $\calU$ is a finite set and $\calI \subseteq 2^\calU$ satisfies the following three properties:
(i) $\emptyset \in \calI$; (ii) For any two sets, $\calA \subseteq \calB \subseteq \calU$, if $\calB \in \calI$, then $\calA \in \calI$; (iii) For any two sets, $\calA, \calB \in \calI$, if $\vert \calB \vert > \vert \calA \vert$, then there exists $e \in \calB \setminus \calA$ such that $\calA \cup \{e\} \in \calI$.
\end{defn}

The following result establishes that the cost function of the optimization problem in \eqref{eq:opt_energy_epsilon_mod} is $\alpha$-supermodular and that its constraint set is a matroid. 
\begin{prop} \label{prop:submodular_proof}
The objective function in \eqref{eq:opt_energy_epsilon_mod}, $E(\calT) \triangleq \trace { (\matW_{\calS(\calT\cup\calT_0)})^{-1} }$ is $\alpha$-supermodular with $\alpha$ satisfying 
\begin{equation} \label{eq:submodular_proof}
    \alpha \geq \cfrac{\lambda_{\min}(\matW_{\calG_0})}{\lambda_{\max}(\matW )} > 0,
\end{equation}
 where $\matW \triangleq \sum_{k=1}^{K} \matA^{k-1} \matB \matB^{\T} (\matA^{k-1})^{\T}$ and $\lambda_{\max}(\cdot), \lambda_{\min}(\cdot)$ are the largest and smallest eigenvalues. Moreover, pair \textcolor{black}{($\calV \setminus \calT_0$, $\lc\calT:\;\calS(\calT \cup \calT_0)\in\Phi \rc$)} is a matroid.
\end{prop} 
\begin{proof}
See \Cref{app:submodular_proof}.
\end{proof}

From \Cref{prop:submodular_proof}, we have the following guarantee for \Cref{alg:greedy_algo_epsilon}. Its proof follows along the lines of \cite[Theorem 1]{Kondapi_icc_2024}; the only differences are that $\tilde{E}$ in~\cite{Kondapi_icc_2024} is replaced with $\trace{\hspace{-0.5mm}\matW_{\calG}^{-1}}$ and $n/\epsilon$ is replaced with $\trace{\hspace{-0.5mm}\matW_{\calG_0}^{-1}}$.
\begin{theorem}
\label{thm:Act_sch_guarantee}
    Let $\beta = \min \lb \frac{\alpha}{2}, \frac{\alpha}{1+\alpha} \rb$, where $\alpha$ is the submodularity constant of the cost function in \eqref{eq:opt_energy_epsilon_mod}. The cost function corresponding to the set $\calG$ returned by \Cref{alg:greedy_algo_epsilon} satisfies
\begin{equation}\label{eq:Act_sch_guarantee_2}
\trace{\hspace{-0.5mm}\matW_\calG^{-1}} \leq \ls 1-\beta \rs \trace{\matW_{\calG_0}^{-1}} + \beta E^*,
\end{equation} 
where $E^* \leq \trace{\matW_{\calG_0}^{-1}}$ is the cost corresponding to the optimal solution of~\eqref{eq:opt_energy_epsilon_mod}. 
\end{theorem} 

\Cref{thm:Act_sch_guarantee} shows that the cost function evaluated at the output of \Cref{alg:greedy_algo_epsilon} is upper bounded by a convex combination of $\trace{\matW_{\calG_0}^{-1}}$ and $E^*$. 
Therefore, to minimize this bound, we next discuss how to choose $\calG_0$ leveraging \Cref{alg:LI_Schedule_FullB}. 
To elaborate, in step 3 of \Cref{alg:LI_Schedule_FullB}, 
there are potentially multiple ways of choosing $l(i)$ LI columns. Among these choices, we can select a subset of columns that are not only linearly independent of the columns chosen till that iteration, 
but also minimize $\trace{\matW_{\calG_0}^{-1}}$. In particular, we can minimize $\trace{\lb\matW_\calS + \eps\matI\rb^{-1}}$ in a greedy fashion to find a $\calG_0$ that ensures controllability and also reduces $\trace{\matW_{\calG_0}^{-1}}$. The term $\eps\matI$ ensures that the inverse is computable. Thus, the procedure described in \Cref{alg:greedy_algo_epsilon} can be incorporated in \Cref{alg:LI_Schedule_FullB}  to find $\calG_0$. We can iteratively pick $l(k)$ columns one-by-one from $\matA^{k}\matB$ that greedily minimizes $\trace{\lb\matW_\calS + \eps\matI\rb^{-1}}$, and thereby find a schedule $\calG_0$ that not only guarantees controllability ($\rank{\matW_{\calG_0}} = n$) but also reduces $\trace{\matW_{\calG_0}^{-1}}$.

Furthermore, suppose that, for some schedule $\calS$, we have $\rank{\matW_{\calS}}=R$, which implies $\matW_{\calS}$ has $R$ nonzero eigenvalues $\{\lambda_i\}_{i=1}^{R}$. Then, for $\epsilon > 0$, the objective function's value~is
\begin{align} \label{eq:approx_obj_fn}
	\trace{(\matW_{\calS} + \epsilon \matI)^{-1}} &= \sum_{i=1}^{R} \frac{1}{\lambda_i+\epsilon} + {\frac{n-R}{\epsilon}}.
\end{align}
We see that the term $(n - R)/\epsilon$ acts as a penalty on the cost when $\matW_{\calS}$ is rank deficient. Due to this, the greedy algorithm picks linearly independent columns in each iteration when $\eps$ is sufficiently small, as asserted in the following proposition. The proof is similar to \cite[Proposition 1]{Kondapi_icc_2024} and is omitted.
\begin{prop} \label{prop:greed_finds_LI}
    Suppose there exist two schedules $\tilde{\calS}, \hat{\calS} \in \Phi$ such that $\rank{\matW_{\tilde{\calS}}}=R$ and $\rank{\matW_{\hat{\calS}}}=R+1$. Then, there exists $\epsilon^*>0$ such that for any $\epsilon<\epsilon^*$, the greedy algorithm chooses $\hat{\calS}$ over $\tilde{\calS}$.
    \end{prop}

Recall that \Cref{prop:MaxLIColumns} implies the existence of LI columns that satisfy the sparsity constraint and ensures controllability when $\rank{\matB}=n$. Hence, we can use \Cref{alg:LI_Schedule_FullB} and \Cref{alg:greedy_algo_epsilon} to find $\calG_0$ that ensures controllability and also reduces $\trace{\matW_{\calG_0}^{-1}}$. 

\emph{Remark:} The actuator scheduler algorithms presented in \Cref{sec:greedy_sch,sec:FullRankB_CtrlGuarantee} do not depend on the initial state $\vecx(0)$. However, the initial state information is required to compute the inputs $\{\vecu(k)\}_{k=0}^{k=K-1}$. Specifically, when the initial state is known, we can compute the inputs using the least squares method using the relation, 
\begin{equation}\label{eq:control_eq}
	\vecx_f - \matA^K \vecx(0) = \matR \vecu_{(K)},
\end{equation}
where $\vecu_{(K)} \triangleq 
 \begin{bmatrix}
   \vecu(0)^{\T} & \vecu(1)^{\T} & \ldots & \vecu(K-1)^{\T} 
\end{bmatrix}
^{\T}$ is a concatenated set of input vectors and is said to be \emph{piecewise sparse}. Also, $\matR$ is the controllability matrix obtained by considering $\calS_i=\{1,2,\ldots,m\}$ for $i = \{0,1,\ldots,K-1\}$ in \eqref{eq:control_con}. 
The support of $\vecu(k)$ is $\calG_k$, where $\calG=(\calG_0, \calG_1, \ldots, \calG_{K-1})$. A solution to \eqref{eq:control_eq} exists, since $\rank{\matR_\calG} = n$.


When the initial state $\vecx(0)$ is unknown, we need to collect measurements to first estimate $\vecx(0)$ and then design the inputs $\vecu(k)$. We address this issue in the next section.

\section{Initial State Unknown}
\label{sec:initstate_unknown}
As mentioned above, when the initial state $\vecx(0)$ is unknown, it must be estimated using measurements of the system state. We consider measurements $\vecy(k)$ at time $k$ of the form
$\vecy (k) = \matC \vecx (k),$ 
where $\matC\in\mathbb{R}^{p\times n}$ is the measurement matrix. We denote the LDS described by the above measurement equation along with \eqref{eq:state} by the tuple $(\matA,\matB,\matC)$. 

We use the measurements $\vecy(k)$ with arbitrary sparse inputs applied over $\tilde{K}$ time steps to solve for $\vecx(0)$. As the inputs are known, we can remove their contribution to the measurements. This is equivalent to applying all zero inputs over $\tilde{K}$ time steps, and the corresponding measurement sequence is
\begin{equation}\label{eq:solve_x0}
[\vecy(0)^{\T} \vecy(1)^{\T} \ldots \vecy(\tilde{K}-1)^{\T}]^{\T} = \matO^{(\tilde{K})} \vecx(0).
\end{equation}
Here, $\matO^{(\tilde{K})}\in\bbR^{\tilde{K}p\times n}$ is the observability matrix of the LDS: 
\begin{equation}
\matO^{(\tilde{K})} \triangleq
\begin{bmatrix}
\matC^{\T} & (\matC \matA)^{\T} & (\matC \matA^2)^{\T}  \ldots (\matC \matA^{\tilde{K}-1})^{\T}
\end{bmatrix}^{\T}.
\end{equation}

Whether \eqref{eq:solve_x0} admits a unique solution depends on the observability of the LDS $(\matA,\matB,\matC)$. So, in the following, we look at two settings: observable and unobservable systems.
 
\emph{Observable Systems:} If the LDS is observable, we can uniquely determine $\vecx(0)$ from \eqref{eq:solve_x0}. After computing the initial state, we can use the algorithms developed in \Cref{sec:greedy_sch,sec:FullRankB_CtrlGuarantee} to compute the actuator schedule and inputs using the least squares method to drive the LDS to the desired state.

Similar to the bound in \eqref{eq:sparse_cntrl_bounds}, the bounds on $\tilde{K}$ for the system to be observable are given by~\cite[Section 6.3.1]{linear_sys} 
\begin{equation}
	\frac{n}{\rank{\matC}} \leq \tilde{K} \leq \min(q, n-\rank{\matC}+1) \leq n,
\end{equation}
with $q$ as defined in \eqref{eq:sparse_cntrl_bounds}. Combining the bound with \eqref{eq:sparse_cntrl_bounds}, the total time steps required to drive the LDS to the desired state using $s$-sparse inputs changes to 
 \begin{multline}
 \frac{n}{\rank{\matC}} + \frac{n}{s} \leq K^* \leq \min(q, n-\rank{\matC}+1) \\
  + \min\lb q \left\lceil \frac{\rank{\matB}}{s} \right\rceil, n-s+1 \rb.
\end{multline}
So, for any observable and $s$-sparse controllable LDS, the maximum number of time steps required to drive the LDS to the desired state is $K^*\leq2n$ when $\vecx(0)$ is unknown. 

We note that transmitting $p$ measurements at each time step from the sensor(s) to the controller incurs a communication overhead. To address this, we can employ \emph{sensor scheduling}, where we transmit at most $\tilde{s}$ measurements to the controller/observer in each time step. Here, the goal is to choose a subset of sensors $\tilde{\calS}_{k} \subseteq \{1,2,\ldots,p\}$ at time $k \in \{0,1,\ldots,\tilde{K}-1\}$ with $\vert \tilde{\calS}_{k} \vert < \tilde{s}$, such that the system remains observable, i.e.,  $\rank{\matO_{\tilde{\calS}}^{(\tilde{K})}} = n$, where
\begin{equation*}
    \matO_{\tilde{\calS}}^{(\tilde{K})}= 
    \begin{bmatrix}
        (\matC^\T)_{\tilde{\calS}_{0}} \! &  \! \matA^\T(\matC^\T)_{\tilde{\calS}_{1}} \! & \! \ldots &  \!(\matA^{\tilde{K}-1})^\T(\matC^\T)_{\tilde{\calS}_{\tilde{K}-1}}
    \end{bmatrix}^{\T}.
\end{equation*}
The sensor scheduling and actuator scheduling problems are duals of each other. So, we can apply results from~\cite{gjoseph2020controllability} to characterize the minimum number of time steps needed to observe and control the system. From the equivalence to sparse controllability, we need $\tilde{s} \geq n-\rank{\matA}$. We can then repurpose \Cref{alg:LI_Schedule_FullB} to find a tuple of subsets $\tilde{\calS}$  of rows of $\matO^{(\tilde{K})}$ that guarantees observability (i.e., $\rank{\matO_{\tilde{\calS}}^{(\tilde{K})}}=n$) when $\rank{\matC}=n$. In this case, $\tilde{K}$ can be bounded as
\begin{equation*}
    \frac{n}{\tilde{s}} \leq \tilde{K} \leq \min\lb q \left\lceil \frac{\rank{\matC}}{\tilde{s}} \right\rceil, n-\tilde{s}+1 \rb \leq n,
\end{equation*}
So, in the first $\Tilde{K}$ time steps, we can observe the system and compute $\vecx(0)$. In the next $K$ time steps, we can apply $s$-sparse control inputs to drive the system to the desired state. Next, we present a result that guarantees the observability and controllability of the LDS with only a few active sensors and actuators per time step.

\begin{cor}
    If $\rank{\matB} = n, \rank{\matC} = n$, and $s, \tilde{s} \geq \max\{1,n-\rank{\matA}\}$, then the LDS $(\matA,\matB,\matC)$ can be driven from any unknown initial state $\vecx(0)$ to any target state $\vecx_f$ in $K^* = K+\tilde{K}$ time steps, where $\tilde{K} \geq \lceil\frac{n}{\tilde{s}}\rceil$ and $K \geq \lceil\frac{n}{s}\rceil$.
\end{cor}
\begin{proof}
    The proof follows from \Cref{thm:FullB_Ctrl_Guarantee}. We can find a sensor schedule that recovers $\vecx(0)$ in $\tilde{K} \geq \lceil\frac{n}{\tilde{s}}\rceil$ time steps using \Cref{alg:LI_Schedule_FullB}, and similarly find an actuator schedule to drive the state from $A^{\tilde{K}}\vecx(0)$ (without loss of generality (w.l.o.g.), no inputs are applied during the initial state estimation) to $\vecx_f$ in $K \geq \lceil\frac{n}{s}\rceil$ time steps.
\end{proof}
This shows that $\Tilde{s}$ measurements and $s$-sparse control inputs per time step are sufficient to drive the LDS to any target state $\vecx_f$, even when the initial state $\vecx(0)$ is unknown.

\emph{Unobservable Systems:} Suppose the LDS is unobservable, but $s$-sparse controllable. Then, the solution $\hat{\vecx}(0)$ to \eqref{eq:solve_x0} yields an estimate, $\hat{\vecx}(\tilde{K}) = \matA^{\tilde{K}}(\vecx(0) + \vecz)$  (w.l.o.g. no inputs are applied during the initial-state estimation), where $\vecz$ belongs to the null space of $\matO^{(\tilde{K})}$. Now, we can design the $s$-sparse control inputs to drive the system from $\hat{\vecx}(\tilde{K})$ to $\vecx_f$ for the next $K$ time steps. Then, at time $K^* = \tilde{K}+K$, the actual state of the system  is $\vecx(K^*) = \vecx_f - \matA^{K^*} \vecz$, and the output of the system is $\vecy(K^*) = \matC\vecx_f$. Hence, we conclude that for an unobservable LDS with unknown initial state, when $\rank{\matB}  = n$, the LDS can be driven to the desired \emph{output} state $\matC \vecx_f$ using \Cref{alg:greedy_algo_epsilon}. 

In the next section, we consider the problem of reaching a target state $\vecx_f$ using sparse control inputs for an LDS with process and measurement noise.

\section{Sparse Control of Noisy Systems} \label{sec:Nsy_Sys}
The algorithms presented in \Cref{sec:FullRankB_CtrlGuarantee} and \Cref{sec:greedy_sch} provide a sparse actuator schedule that can be used to drive the LDS from any given initial and final state, under a noiseless setting. If the same schedule is employed in LDSs affected by noise, the resulting trajectory may stray away from the desired final state. So, we need feedback control to compensate for both process noise and measurement noise. Specifically, we consider a system with additive Gaussian noise in the state evolution and measurement equations:
\begin{align}
\vecx (k+1) &= \matA \vecx (k) + \matB \vecu(k) + \vecv(k) \label{eq:state_noise} \\
\vecy (k) &= \matC \vecx (k) + \vecw(k), \label{eq:obsvn_noise}
\end{align}
where the noise terms $\vecv(k) \sim \calN(\zero$, $\bs{\Sigma}_{v})$, $\vecw(k) \sim \calN(\zero, \bs{\Sigma}_{w})$ are mutually independent and are drawn from known distributions. We consider the initial state $\vecx(0)$ to be known. Our results also hold when $\vecx(0)$ is drawn from a Gaussian distribution $\calN(\vecx_0,\bs{\Sigma}_0)$ independent of the noises. 

We aim to reach a target state $\vecx_f$ using sparse control inputs at each time instant. We cannot determine the number of time steps required for the system to reach $\vecx_f$ as done in the previous section, due to noise. So, we cannot consider a finite horizon linear quadratic Gaussian problem similar to~\cite{Chamon2022ApproxSuperMatroid}; instead, we consider a Kalman filtering-based approach 
to estimate $\vecx(k)$ and determine the sparse inputs to be applied to drive the next state of the system as close to the target state as possible. 
Let the minimum mean square error (MMSE) estimate of the state $\vecx(k)$ produced by the Kalman filter based on the measurements $\vecy(0), \vecy(1), \ldots, \vecy(k)$ be denoted by $\hat{\vecx}(k)$. 
Using this estimate, we find the input $\vecu(k)$ such that $\matA \hat{\vecx}(k) + \matB \vecu(k)$ is close to $\vecx_f$. That is, we solve the problem 
\begin{equation}
\label{eq:kf_based_opt}
\min_{\vecu(k)} \; \Vert \vecx_f - \matA \hat{\vecx}(k) - \matB \vecu(k) \Vert \;
\text{s.t.} \; \Vert \vecu(k) \Vert_0 \leq s.
\end{equation}
One approach to solve \eqref{eq:kf_based_opt} is to use the orthogonal matching pursuit (OMP) algorithm~\cite{Pati_1993_OMP}, which adds an element that minimizes the objective to the support of $\vecu(k)$ at each iteration, in a greedy fashion. The pseudo-code for the overall solution that combines Kalman filtering with the OMP algorithm is summarized in \Cref{alg:sparse_control_noise}. The computational complexity of \Cref{alg:sparse_control_noise} is $\calO(n \max\{n^2, p^2,m,s^3\}+s^3m+p^3)$.

\begin{algorithm}[t]
	\caption{Algorithm to compute sparse input at time $k$}
	\begin{algorithmic}[1]
		\Require Desired state $\vecx_f$; System matrices $\matA,\matB ,\matC$; LDS output $\vecy(k) \in \mathbb{R}^p$, previous input $\vecu(k-1)$; previous state estimate $\hat{\vecx}(k-1)$; its covariance  $\matP_{k-1}$; Noise covariance $\matSigma_v,\matSigma_w$; Sparsity $s$ 
		\Statex\emph{\#Kalman filter}
		\State Compute prediction's covariance: $\hat{\matP}=\matA \matP_{k-1} \matA^{\T} + \bs{\Sigma}_{v}$
	    \State Compute Kalman gain $\matK = \hat{\matP} \matC^{\T} \lb \matC \hat{\matP} \matC^{\T} + \matSigma_w \rb^{-1}$
	    \State Estimate LDS state $\hat{\vecx}(k) = \matA \hat{\vecx}(k-1)+\matB\vecu(k-1) + \matK \lb \vecy(k) - \matC \matA \hat{\vecx}(k-1) \rb$ \label{line:kalman_est}
	    \State Compute estimate's covariance $\matP_{k} = (\eye- \matK \matC)\hat{\matP}$ \label{line:cov_update}
	    \Statex\emph{\#Sparse recovery}
		\State Solve for $\vecu(k)$ using OMP from \eqref{eq:kf_based_opt}
		\Ensure Input $\vecu(k)$; state estimate $\hat{\vecx}(k)$; its covariance $\matP_{k}$
	\end{algorithmic} \label{alg:sparse_control_noise}
\end{algorithm}

The OMP algorithm possesses the following theoretical guarantee on its solution~\cite{elad_book}
\begin{equation} \label{eq:omp_guarantee}
\Vert \vecx_f - \matA \hat{\vecx}(k) - \matB \vecu(k)  \Vert^2 \leq \xi^s(\matB) \Vert \vecx_f - \matA \hat{\vecx}(k) \Vert^2,
\end{equation} 
where $\xi(\matB)$ is the so-called \emph{universal decay-factor} of $\matB$, defined as
\begin{equation}\label{eq:defn_delB}
\xi(\matB) = 1-\inf_{\{\vecv \in \mathbb{R}^n | \Vert \vecv \Vert = 1 \} } \max_{1\leq j\leq m } \frac{\vert \matB_j^{\T} \vecv \vert^2}{\Vert \matB_j \Vert^2}.
\end{equation}
For example, the quantity $\xi(\matB) = 1-1/n$ when $\matB = \matI$. The above guarantee says that the norm of the residual reduces exponentially with $s$. 
Also, the OMP algorithm increments the support of $\vecu(k)$ such that the corresponding subset of columns $\matB$ are linearly independent. This feature of OMP leads to the following guarantee for our algorithm.
\begin{theorem}
\label{thm:noise_algo_guarantee}
Consider the LDS defined in \eqref{eq:state_noise} and \eqref{eq:obsvn_noise}. Let $(\matA, \mathbf{\Sigma}_v^{1/2})$ be controllable and $(\matA, \matC)$ be observable. Suppose we apply the input sequence produced by \Cref{alg:sparse_control_noise} (with the desired state $\vecx_f$) to the system. If $s> \max \{ -\frac{\log(2\Vert \matA \Vert^2)}{\log(\xi(\matB))}, 0 \}$ and $\rank{\matB}=n$, then for any $\eta > 0$, there exists $k_0$ such that for all $k \geq k_0$, the state sequence $\vecx(k)$ thus obtained satisfies 
\begin{multline}\label{eq:noise_algo_guarantee}
 \mathbb{E}\{ \Vert \vecx(k) - \vecx_f \Vert^2\}  \leq  \eta+\frac{2\xi^s(\matB)\Vert (\matA -\eye)\vecx_f\Vert^2}
{1-2\xi^s(\matB)\Vert \matA \Vert^2}\\+\frac{\trace{\bs{\Sigma}_v}+(1+3\xi^s(\matB)) \trace{\matA \matP \matA^{\T}}  }{1-2\xi^s(\matB)\Vert \matA \Vert^2}.
\end{multline}
Here, $\xi(\matB)$ is as defined in \eqref{eq:defn_delB} and $\matP \in \mathbb{R}^{n \times n}$ satisfies the following system of equations:
\begin{align} 
	\matP &= \matS - \matS \matC^T \left( \matC \matS \matC^T + \bs{\Sigma}_w \right)^{-1} \matC\matS \label{eq:KF_cnvg_1}\\
	\matS &= \matA \matP \matA^T + \bs{\Sigma}_v  .\label{eq:KF_cnvg_2}
\end{align} 
\end{theorem}

\begin{proof}
See \Cref{app:KF_based_alg}.
\end{proof}

We note that the last term in the upper bound in \eqref{eq:noise_algo_guarantee} depends on the process and measurement noise covariance matrices $\bs{\Sigma}_v$ and $\bs{\Sigma}_w$. This term characterizes the impact of process noise and Kalman filter estimation error. If $s=n$, we have $\bbE\{\lV \vecx(k)-\vecx_f \rV^2\} = \trace{\bs{\Sigma}_v + \matA^{\T}\matP_{k-1}\matA}$ from \eqref{eq:step_main1} in \Cref{app:KF_based_alg}, and $\matP_{k-1}$ converges to $\matP$ asymptotically. Note that this expression also serves as a lower bound on the MSE at any $s$. We observe that when $\vecx_f$ is an equilibrium point of the autonomous system $\matA$ (e.g., $\vecx_f=\zero$), the upper bound in \eqref{eq:noise_algo_guarantee} only have terms related to the estimation error and process noise, and it is a scaled version of $\trace{\bs{\Sigma}_v + \matA^{\T}\matP\matA}$, the MSE when $s=n$. Hence, we can expect the performance of \Cref{alg:sparse_control_noise} to be near optimal. 

If, for some $k$, we have $\hat{\vecx}(k) = \vecx_f$ and $\vecx_f$ is not an equilibrium point, then $\vecu(k)$ must satisfy $\matB \vecu(k) = (\matI-\matA)\vecx_f$. Due to the sparsity constraint on $\vecu(k)$, it may not be possible to remain in state $\vecx_f$ by applying sparse inputs, even in the noiseless case.\footnote{In the noiseless setting, \cite{icassp_paper} presents a test to check whether $\matB\vecu(k) = (\matA-\matI)\vecx_f$ admits a sparse solution $\vecu(k)$.} The second term in \eqref{eq:noise_algo_guarantee} thus captures the algorithm's inability to stabilize the system at $\vecx_f$ by applying sparse inputs; the term decays exponentially with $s$. 
Our simulation results in \Cref{sec:simulations} corroborate these observations and show that \eqref{eq:noise_algo_guarantee} captures the behavior of the MSE incurred by \Cref{alg:sparse_control_noise} accurately. 

\emph{Remark:} \Cref{alg:sparse_control_noise} can be readily extended to the case where a regularization term on $\ell_2$ norm of the input $\vecu(k)$ or $\vecx(k)$ is explicitly imposed in \eqref{eq:kf_based_opt}. If $\matQ \succeq \zero$ and $\matR \succ \zero$ are the regularizers for the state and input vectors, respectively, then we can consider the following optimization problem, 
\begin{equation}
    \underset{\vecu(k)}{\min} \lV \matB^{\T}\matQ\lb\vecx_f-\matA\hat{\vecx}(k)\rb - \Tilde{\matB}\vecu(k) \rV \text{ s.t. } \lV\vecu(k)\rV_0 \leq s,
\end{equation}
where $\tilde{\matB} = \matB^{\T}\matQ\matB + \matR$. We note that the above reduces to \eqref{eq:kf_based_opt} when $\matQ = \matI$ and $\matR = \zero$. Also, similar to~\cite{Vafaee2023LargeSensorNwk}, sensor scheduling can be used to reduce the cost of state estimation. 

In the next section, we evaluate the performance of the algorithms developed above via simulations and compare them against our theoretical bounds. 

\section{Simulation Results}
\label{sec:simulations}

In this section, we illustrate the performance of the time-varying and sparse actuator scheduling algorithm developed in this work by applying it to an LDS whose transfer matrix $\matA$ is generated from an Erd\H{o}s-Renyi random graph. An Erd\H{o}s-Renyi random graph consists of $n$ vertices. An edge independently connects a pair of vertices with probability {$p_e = \frac{2 \ln n}{n}$}, a setting used to model real-world networked systems as well as social networks~\cite{Tzoumas_2016_CtrlEffort, Liu_2011_CtrlCmxNwk}. 
This ensures that the generated Erd\H{o}s-Renyi graph is connected with high probability as $n$ gets large. The transfer matrix of the LDS is then generated from the graph as $\matA = \matI - \frac{1}{n} \matL$ as described in~\cite{jadbabaie2018deterministic, Saber_2007_Consesus}, where $\matL$ is the graph Laplacian. Here, $\matA$ is a row stochastic matrix, so we have $\lV \matA \rV$ = 1. Also, when $p_e=0$, we have $\matA = \matI$ and when $p_e=1$, we have $\matA = \frac{1}{n}\one\one^{\T}$. We consider the entries of the input matrix $\matB$ are i.i.d. (either distributed uniformly over $[0,1]$ or $\calN(0,1)$)~\cite{icassp_paper}. In the case where $m=n$, we also consider $\matB=\matI$.


\begin{table}[t]
    \centering
    \caption{The average control energy, $\log\{\trace{\matW_\calS^{-1}}\}$, vs. sparsity level $s$. \texttt{RB$n$-greedy} is the proposed \Cref{alg:greedy_algo_epsilon}, while the other two are algorithms in~\cite{Ballotta_2024_PtwiseSch}. The existing algorithms fail to obtain a sparse controllable schedule when $s=2$.}
    \begin{tabular}{|c|c|c|c|c|}
         \hline
         Algorithm & $s=2$ & $s=3$ & $s=4$ & $s=5$ \\ \hline
         RB$n$-greedy & $10.9535$ & $6.1344$ & $3.8603$ & $2.67244$ \\ \hline
         $s$-sparse greedy+MCMC & fails & $9.99705$ & $4.27217$ & $3.34841$ \\ \hline
         $s$-sparse greedy & fails & $10.0018$ & $6.80099$ & $5.59268$ \\ \hline
    \end{tabular}
    \label{tbl:ActComparisonCE}
\end{table}

\begin{figure}[t]
    \centering
        \includegraphics[width = 0.9\linewidth]{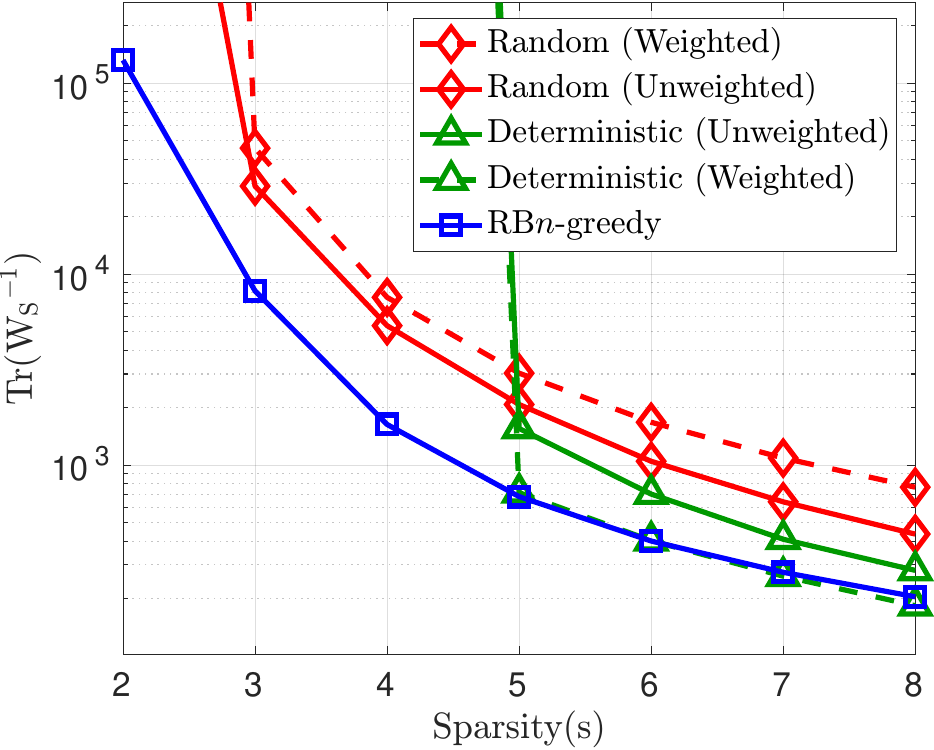}
        \caption{$\trace{\matW_\calS^{-1}}$ as a function of sparsity level $s$, averaged over $100$ independent trials with $n=100, m=n, K=50$.} 
        \label{fig:ActComparison}
\end{figure}

\subsection{Noiseless LDS}
First, we compare the RB$n$-greedy actuator scheduling algorithm, \Cref{alg:greedy_algo_epsilon}, with the schedulers presented in~\cite{Ballotta_2024_PtwiseSch}. We consider the system\footnote{Code for generating the results in this section can be found here: \url{https://github.com/KondapiPraveen/Sparse-Controllability}} 
$(\matA,\matB)$ where $n=20$ with $p_e = 0.89$ and the control horizon is $K=\lceil\frac{n}{s}\rceil$. In \Cref{tbl:ActComparisonCE}, we show the behavior of the average control energy $\trace{\matW_{\calS}^{-1}}$ as a function of the sparsity level $s$ for \Cref{alg:greedy_algo_epsilon}, and the existing $s$-sparse greedy, $s$-sparse greedy+MCMC from~\cite{Ballotta_2024_PtwiseSch}. It can be seen that for $s=2$, the  $s$-sparse greedy and $s$-sparse greedy+MCMC schedulers fail to ensure the controllability of the system. Further, \Cref{alg:greedy_algo_epsilon} outperforms $s$-sparse greedy+MCMC and $s$-sparse greedy algorithms, with the gap reducing as $s$ increases. \Cref{alg:greedy_algo_epsilon} always provides a schedule that ensures controllability as long as $\rank{\matB}=n$, due to \Cref{thm:FullB_Ctrl_Guarantee}.

Now, we compare \Cref{alg:greedy_algo_epsilon} with schedulers adapted from~\cite{jadbabaie2018deterministic}. The algorithms in \cite{jadbabaie2018deterministic} are developed under an \emph{average} sparsity constraint across the inputs, i.e., the individual inputs need not be sparse. 
Hence, we slightly modify the algorithms in \cite{jadbabaie2018deterministic} so that they conform to the sparsity constraint considered in this work. In \Cref{fig:ActComparison}, we plot  $\trace{\matW_{\calS}^{-1}}$ as a function of $s$ for four algorithms adapted from \cite{jadbabaie2018deterministic}, namely, random-weighted, random-unweighted, deterministic-unweighted, and deterministic-weighted. The random scheduler samples an actuator from the probability distribution given in \cite[Algorithm 6]{jadbabaie2018deterministic} and adds it to the schedule, provided adding it still satisfies the sparsity constraint. The deterministic scheduler picks an actuator that greedily optimizes the objective function used in \cite[Algorithm 1]{jadbabaie2018deterministic}. The weighted schedulers have an additional weight (amplification) associated with the selected actuator, while the weights are set to $1$ for the unweighted schedule (see \cite[Algorithms 2, 6, 7]{jadbabaie2018deterministic}.) For the unweighted random scheduler, the actuators are sampled without replacement. The plot shows that \Cref{alg:greedy_algo_epsilon} generally outperforms existing algorithms modified to satisfy our sparsity constraint, especially at lower sparsity levels.

\begin{figure}[t]
    \centering
        \includegraphics[width = 0.9\linewidth]{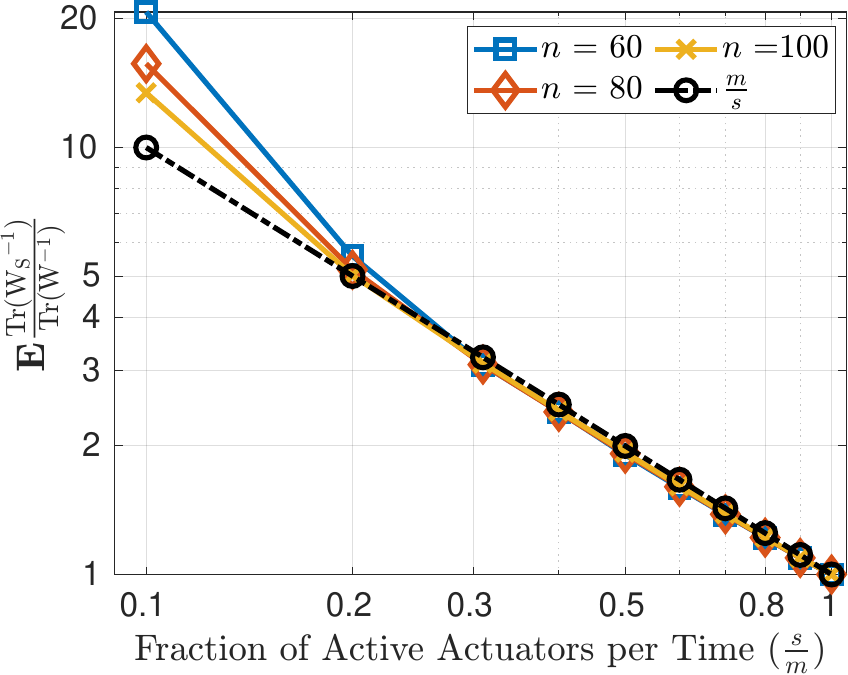}
        \caption{Relative energy cost of imposing piecewise sparsity constraints, $\bbE_{A,B} \frac{\trace{\matW_\calS^{-1}}}{\trace{\matW^{-1}}}$, as a function of the fraction of  active actuators ($\frac{s}{m}$). The relative cost is proportional to $\frac{m}{s}$ as $s$ increases. Here, $n = m$ and the entries of $\matB$ are chosen i.i.d. and uniformly at random over $[0,1]$.} 
        \label{fig:InverseEnergy}
\end{figure}

Next, we illustrate the cost of imposing piecewise sparsity constraints relative to the non-sparse case, by plotting $\rho \triangleq {\bbE_{\matA, \matB} \{ \trace{\matW_\calS^{-1}}}/\trace{\matW^{-1}} \}$ as a function of the fraction of active actuators at each time instant (${s}/{m}$); note that $\rho$ is the ratio of the average control energy with the sparsity constraint to the average control energy in the unconstrained case. We show the behavior in \Cref{fig:InverseEnergy}, with $m=n$, and $n$ varying from $60$ to $100$. We observe an inverse-linear relationship between $s/m$ and $\rho$, meaning that using $s$-sparse control inputs increases the average control energy by a factor proportional to the reciprocal of the fraction of active actuators. We have also observed this trend when $\matB$ is a random matrix with i.i.d. entries drawn from a Gaussian distribution as well as when $\matB$ = $\matI$ (for $n=m$). It is easy to theoretically establish this result for $\matA = \matB = \matI$, since $\matW_\calS = s\matI$ and $\matW=m\matI$ in this case. As $n$ increases, $p_e$ tends to 0 and $\matA$ approaches  $\matI$. Hence, the above result is expected. 

\begin{figure}
    \centering
    \includegraphics[width=0.9\linewidth]{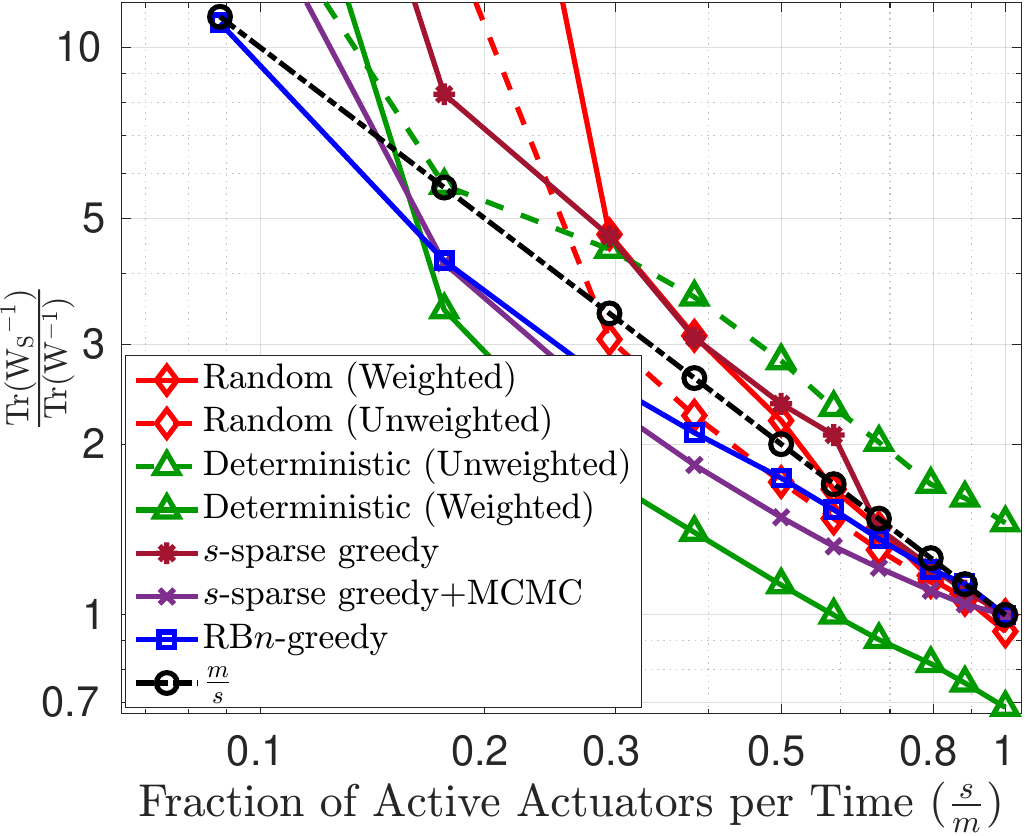}
    \caption{RB$n$-greedy has the smallest relative cost of imposing sparsity constraints than the other schedulers, when $s=3$. Its performance is comparable to $s$-sparse greedy+MCMC. Even on the real-world Zachary's karate social network dataset, the RB$n$-greedy algorithm's cost follows the $m/s$ trend observed in Fig.~\ref{fig:InverseEnergy}.}
    \label{fig:ZacharyInverseEnergy}
\end{figure}

As the last experiment in this subsection, we illustrate the relative average energy to control a real-world social network using sparse inputs. Zachary's karate club~\cite{zachary1977information} is a social network of a karate club comprising $34$ members ($n = 34$), and their pairwise interactions are captured through an adjacency matrix. The transfer matrix $\matA$ is constructed using the Laplacian matrix of this social network, and $\matB = \matI$ corresponding to the direct influence of states of $s$ members at each time step. We plot relative average energy $\rho$ as a function of ${s}/{m}$ for all the actuator schedulers in \Cref{fig:ZacharyInverseEnergy}. From the plot, we can infer that $\rho$ for \Cref{alg:greedy_algo_epsilon} is at most ${m}/{s}$. The deterministic (weighted) and $s$-sparse greedy+MCMC have smaller $\rho$ after $s=6$ because of additional weighting flexibility for actuators and exploration by Monte Carlo experiments, respectively. The performance of most of the actuator schedulers is similar in the high sparsity regime.

\begin{figure}
    \centering
        \includegraphics[width = 0.9\linewidth]{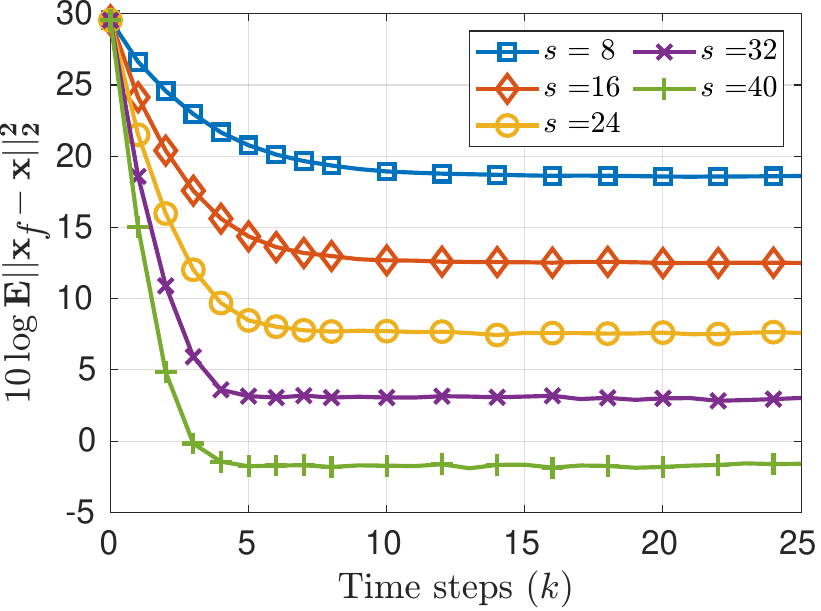}
        \caption{MSE (in dB) as a function of time steps. MSE saturates earlier for larger $s$. For $s \ge 24$, the reduction is MSE is nearly the same for successive $s$.} 
        \label{fig:MSEOverTime1}
\end{figure}

\begin{figure}
    \centering
        \includegraphics[width = 0.9\linewidth]{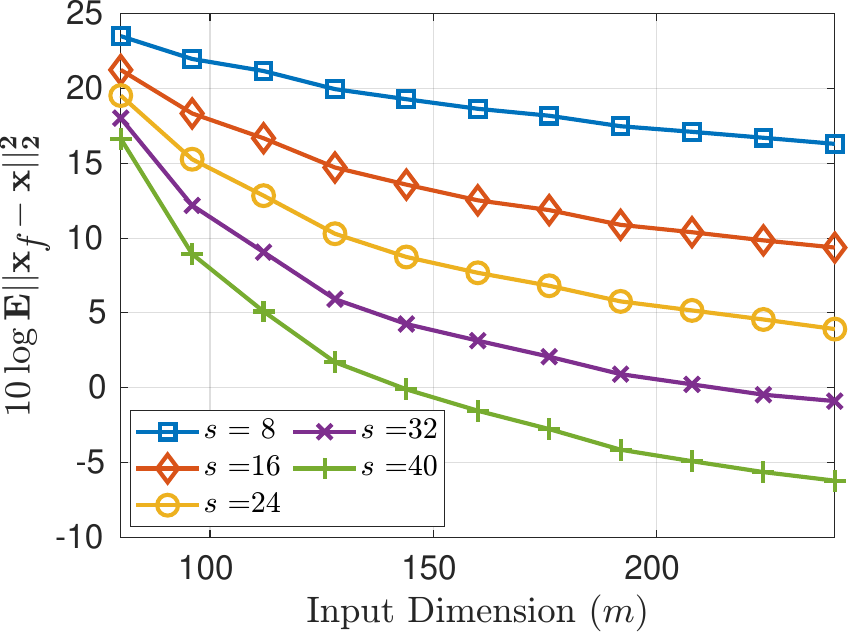}
        \caption{MSE (in dB) as a function of number of actuators ($m$). The gap between the MSE attained at successive sparsity levels increases with~$m$.} 
        \label{fig:MSEOverM}
\end{figure}

\subsection{Noisy LDS}
Now, we demonstrate the performance of \Cref{alg:sparse_control_noise} developed in \Cref{sec:Nsy_Sys} to \emph{track} a target state $\vecx_f$ for the noisy LDS in \eqref{eq:state_noise}, \eqref{eq:obsvn_noise}. The noise covariance matrices are taken as $\bs{\Sigma}_{v} = \bs{\Sigma}_{w} = \sigma^2 \matI$ and the measurement dimension $p=n$. First, we illustrate the MSE over time. In \Cref{fig:MSEOverTime1}, we plot the logarithm of the MSE, computed as $10\log \bbE \lc \lV \vecx_f - \vecx(k) \rV^2 \rc$, as a function of time $k$ with sparsity levels $s$ varying from $8$ to $40$. We consider an LDS with $n=80, m=2n$, the noise variance is taken as $\sigma^2=10^{-4}$ and $\vecx(0) = \zero$. We observe that the MSE saturates earlier for larger $s$, which can be inferred from \eqref{eq:bound_1} in \Cref{app:KF_based_alg}. Here, the saturation level is dominated by the second term in \eqref{eq:noise_algo_guarantee}. Due to the denominator of the second term, initially, we see a larger improvement from successive sparsity levels. Beyond $s=24$, we see that the improvement between successive sparsity levels is roughly constant. This is because the denominator approaches $1$, leading to a smaller effect of increasing $s$ on the change in the saturation level.

Choosing a larger $m$ will improve the MSE for the same sparsity level because $\xi(\matB)$ generally decreases with $m$. In \Cref{fig:MSEOverM}, we illustrate the effect of increasing $m$ on the MSE under the same settings as in \Cref{fig:MSEOverTime1}, but with the total number of actuators ($m$) varied from $80$ to $240$. We observe that having a larger $m$ leads to better tracking for the same sparsity level $s$. Also, as $m$ increases, the reduction in MSE between successive values of $s$ increases. For example, we can achieve an MSE of $10$~dB (which corresponds to reaching within a ball centered at $\vecx_f$ and with radius less than $(0.1 \times \lV \vecx_f \rV)$) in an LDS with $m=95$ inputs if $s=40$, but if $m=110$ or $210$, $s=32$ or $s=16$ actuators suffice, respectively. Increasing $m$ provides the scheduler with more choices for selecting actuators, resulting in a smaller sparsity level being sufficient.

\begin{figure}
    \centering
        \includegraphics[width = 0.9\linewidth]{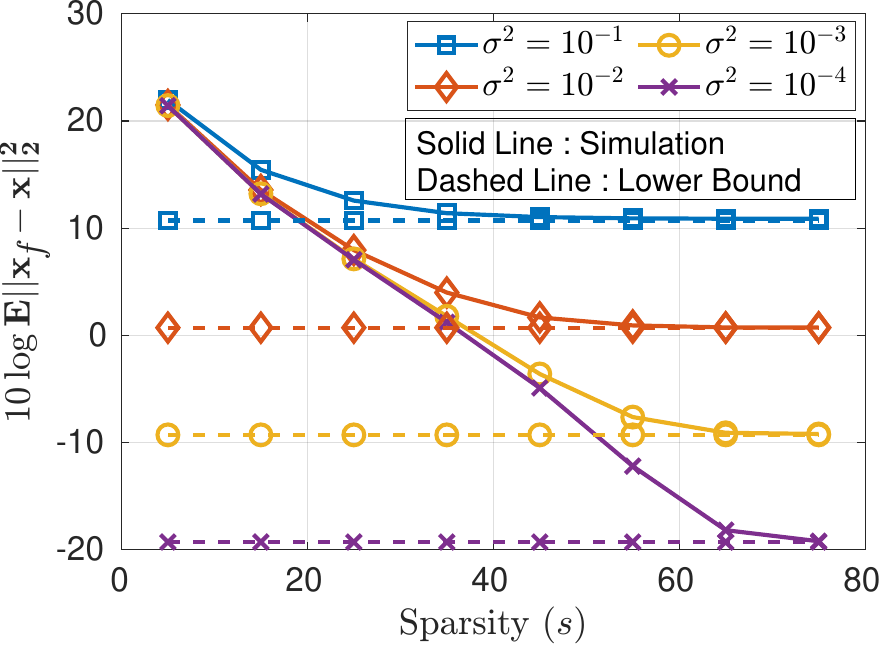}
        \caption{MSE (in dB) as a function of sparsity, $n = 80$ and $m = 160$. The MSE initially decreases linearly in log scale and then saturates.} 
        \label{fig:MSEOverSparsity}
\end{figure}

In \Cref{fig:MSEOverSparsity}, we illustrate the effect of the process and measurement noise variance on the MSE. We plot the \textcolor{black}{MSE (in dB)} at the end of $K=40$ control time steps as a function of sparsity level $s$ for different noise variance values. We consider an LDS with $n=80, m=2n$, and vary the noise variance $\sigma^2$ from $10^{-1}$ to $10^{-4}$. Initially, the MSE decreases exponentially with $s$, and then saturates. The initial exponential decay of the MSE with $s$ is because of the second term in \eqref{eq:noise_algo_guarantee}, which decreases exponentially with $s$. As $s$ increases, the error due to the process noise and the state estimation error (the last term in \eqref{eq:noise_algo_guarantee}) dominates. So, increasing $s$ beyond a certain level does not reduce the MSE significantly. This illustrates that the initial exponential decay and eventual saturation behavior of the MSE are both captured by the upper bound in \Cref{thm:noise_algo_guarantee}. The figure also shows the lower bound on the MSE, namely, $\trace{\bs{\Sigma}_v + \matA^{\T}\matP_K \matA}$, with $\matP_K$ obtained from \Cref{line:cov_update}, which shows that the MSE approaches the lower bound as $s$ increases. 

\begin{figure}
    \centering
        \includegraphics[width = 0.9\linewidth]{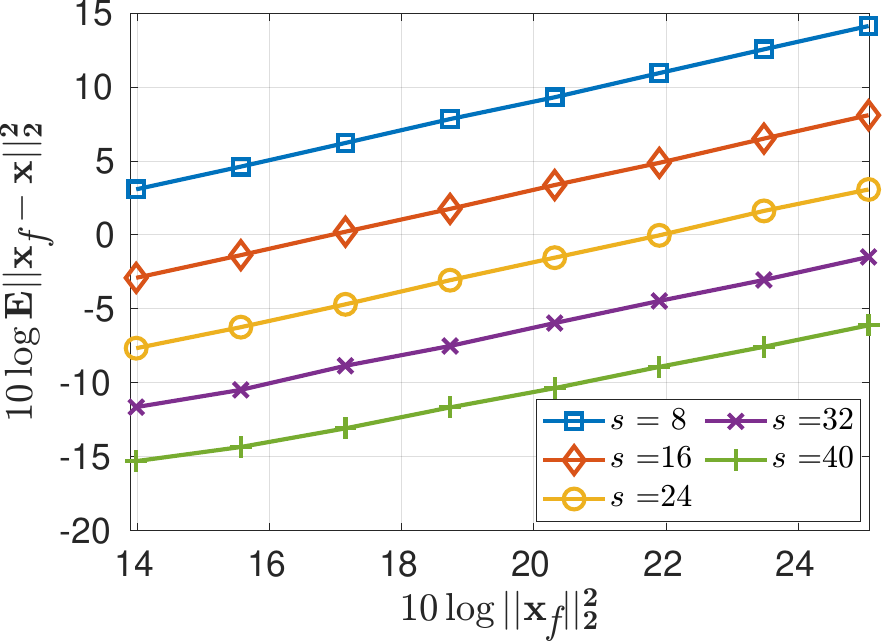}
        \caption{MSE (in dB) as a function of $\log \lV \vecx_f \rV^2$. The dependence is linear and the gap between the curves is roughly constant as $s$ increases, both of which are as expected from the second term in \eqref{eq:noise_algo_guarantee}.} 
        \label{fig:MSEOverNormX_f}
\end{figure}

Next, we illustrate the effect of $\lV \vecx_f \rV$ on the \textcolor{black}{MSE}. In \Cref{fig:MSEOverNormX_f}, we plot the MSE at the end of $K=40$ control time steps as a function of $\lV \vecx_f \rV^2$ (in the $\log$-domain), for $s$ varying from $8$ to $40$. \textcolor{black}{We consider an LDS with $n=80, m=2n$, and the noise variance is $\sigma^2 = 10^{-4}$. As $\lV \vecx_f \rV$ increases, tracking the target state $\vecx_f$ becomes harder because the second term in \eqref{eq:noise_algo_guarantee}  increases with $\lV \vecx_f \rV$. This plot also shows that MSE (in dB) behaves roughly as $-(\text{constant})\cdot s + \log{\norm{\vecx_f}{2}^2}$ as predicted by the second term in \eqref{eq:noise_algo_guarantee}, since the third term is negligible.} 

\begin{figure}
    \centering
        \includegraphics[width = 0.9\linewidth]{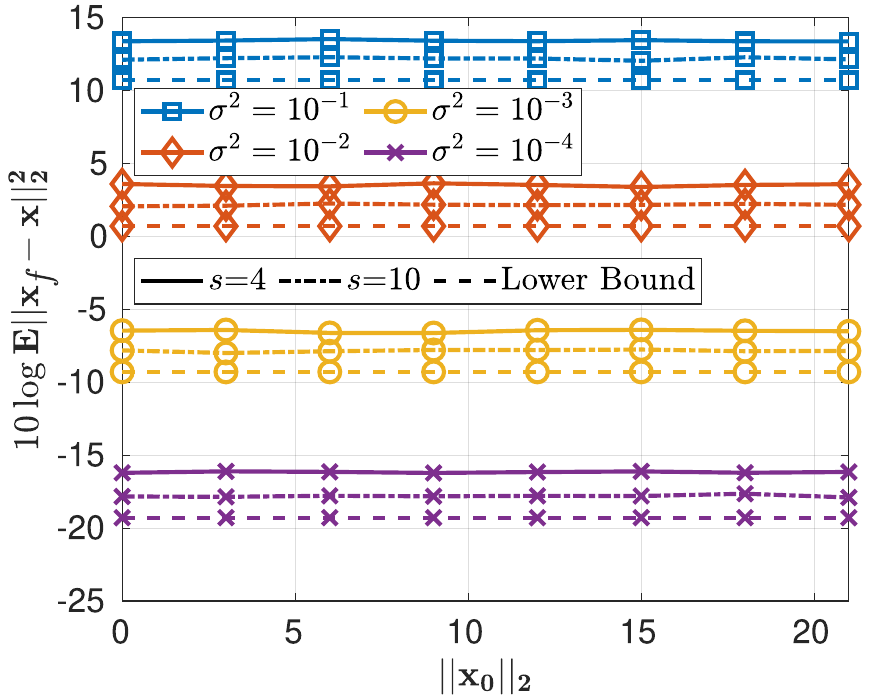}
        \caption{MSE (in dB) as a function of $\lV \vecx_0 \rV$. The MSE is close to the lower bound even for $s=4$.} 
        \label{fig:MSEOverNormX_0}
\end{figure}

Finally, we consider the case when $\vecx_f = \one$ and $\vecx(0) = \vecx_0$ instead of $\zero$. We note that $\vecx_f = \one$ corresponds to consensus among the nodes in the network. In \Cref{fig:MSEOverNormX_0}, we plot the MSE (in dB) at the end of $K=40$ control time steps as a function of $\lV \vecx_0 \rV$ for $\sigma^2$ ranging from $10^{-1}$ to $10^{-4}$, and $s=4$ and $10$. We consider an LDS with $n=80, m=2n$. Since $\matA$ is a row stochastic matrix, $\lV \lb\matA-\matI\rb\vecx_f \rV=0$, the second term in \eqref{eq:noise_algo_guarantee} vanishes. The last term corresponding to noise variance and estimation error, which is independent of $\lV \vecx_0 \rV$, dominates the tracking performance of \Cref{alg:sparse_control_noise}. This pattern is evident in the plot as the noise variance increases: the MSE also increases irrespective of the initial state $\vecx(0)$. We have also observed the same tracking performance for other sparsity levels; we omit the plots to avoid repetition.

\section{Conclusion}
We studied the problem of finding an energy-efficient actuator schedule for controlling an LDS in a finite number of steps, with at most $s$ active inputs per time step. To address this problem, we presented a greedy scheduling algorithm for minimizing the average control energy, $\trace{\matW_\calS^{-1}}$. We presented several interesting theoretical guarantees for the algorithm when the input matrix has full row rank. Simulations showed that the algorithm returns a feasible schedule, with the average energy increasing modestly (compared to the case without sparsity constraints), proportional to the inverse of the fraction of active actuators. Thus, sparse inputs can drive an LDS to a desired state without significantly increasing energy requirements at moderate sparsity levels. We also developed a controller based on the Kalman filter for noisy LDSs and derived an upper bound on its \textcolor{black}{steady-state MSE}. Simulations indicated that increasing the number of available actuators $m$ improves the tracking performance when the sparsity level is fixed. Finally, we showed that, even under sparsity constraints, we can stabilize the system around an equilibrium point in the presence of noise. Future work could look at assessing the sparse controllability under adversaries and under unknown system matrices, especially since data-driven control has been gaining focus recently. 

\crefalias{section}{appendix}
\begin{appendices}
    

\section{Proof of \Cref{thm:FullB_Ctrl_Guarantee}} \label{app:FullB_Ctrl_Guarantee}
    By construction in \Cref{alg:LI_Schedule_FullB}, the elements of set $\calT^{(i)}$ corresponds to linearly independent columns in matrix $\begin{bmatrix}
    \matA^{K-1}\matB & \matA^{K-2}\matB & \ldots & \matA^{i}\matB
\end{bmatrix}$. To ensure controllability, i.e., $\rank{\matR_{\calG_0}}=n$, we need to prove that $\lv \calT^{(0)} \rv = n$ and it is enough to prove the required result for the case of $K = \lceil \frac{n}{s}\rceil$. To this end, it suffices to show that $\lv \calT^{(1)} \rv \geq n-s$ since $\rank{\matB}=n$. The remaining $s$ elements can then be added to $\calT^{(0)}$, corresponding to linearly independent columns from $\matB$, thereby making $\lv \calT^{(0)} \rv = n$. Furthermore, in \Cref{alg:LI_Schedule_FullB}, the number of linearly independent columns selected until any iteration index $i$~is (See \Cref{prop:MaxLIColumns})
    \begin{equation}
        \lv \calT^{(i)} \rv = \min \lc \rank{\matA^{i}\matB}, \lv \calT^{(i+1)} \rv + s \rc.
    \end{equation}
    This is because we only add linearly independent columns in each iteration and $\lv \calI(i) \rv = l(i)$.

    Rewriting $\lv\calT^{(i)}\rv$ using above recursive relation, we get 
    \begin{multline}
        \lv\calT^{(i)}\rv = \min \{ \rank{\matA^i\matB}, \rank{\matA^{i+1}\matB}+s \\ ,\ldots, \rank{\matA^{K-1}\matB}+(K-i-1)s, (K-i)s \}.
    \end{multline}
    Therefore, for $i=1$, we have 
    \begin{multline} \label{eq:T1Cardinality}
        \lv\calT^{(1)}\rv = \min \{ \rank{\matA\matB}, \rank{\matA^{2}\matB}+s ,\ldots, \\ \rank{\matA^{j}\matB}+(j-1)s,\ldots, \\ \rank{\matA^{K-1}\matB}+(K-2)s, (K-1)s \}.
    \end{multline}
    Now, to establish that $ \lv\calT^{(1)}\rv\geq n-s$, it is enough to show that each term inside the minimum operator in \eqref{eq:T1Cardinality} is at least $n-s$. Starting with the last term $(K-1)s$, we have
    \begin{equation} \label{eq:MinInputGuarantee}
        \lb K-1 \rb s = \lb \left \lceil\frac{n}{s}\right\rceil - 1\rb s \geq n-s. 
    \end{equation}
    Now, we show that the remaining terms in \eqref{eq:T1Cardinality} are at least $n-s$. We use Sylvester's rank inequality
   to arrive at
    \begin{equation}
        \rank{\matA^i\matB} \geq \rank{\matA^{i-1}\matB} + \rank{\matA} - n.
    \end{equation}
    Since $s \geq \rank{\matA} - n$, we have
    \begin{equation} \label{eq:AkBinequality}
        \rank{\matA^i\matB} \geq \rank{\matA^{i-1}\matB} - s.
    \end{equation}
    Applying \eqref{eq:AkBinequality} recursively $i-2$ times, we derive
    \begin{equation}
        \rank{\matA^i\matB} \geq \rank{\matA\matB} - (i-1)s.
    \end{equation}
    Finally, since $\rank{\matA\matB} = \rank{\matA}$ and $\rank{\matA} \geq n-s$, we deduce that
    \begin{equation} \label{eq:AkB_Lower_Bnd}
        \rank{\matA^i\matB} + (i-1)s \geq n-s.
    \end{equation}
    
    Combining \eqref{eq:AkB_Lower_Bnd}, \eqref{eq:MinInputGuarantee} and \eqref{eq:T1Cardinality}, we get $\lv \calT^{(1)} \rv \geq n-s$. Hence, $\lv\calT^{(0)}\rv=n$, or  $\rank{\matR_{\calG_0}}=n$, completing the proof. \hfill \qedb

\section{Proof of \Cref{prop:submodular_proof}}\label{app:submodular_proof}
The proof follows along the lines of \cite[Proposition 2]{Kondapi_icc_2024}, except that we define $\matM_{\emptyset} = \matW_{\calG_0} \succ 0$ instead of $\matM_{\emptyset} = \epsilon \mathbf{I}$ and we modify the sparsity constraint in \cite[Eq. (25)]{Kondapi_icc_2024} to $s - |\{j: (k,j) \in \mathcal{T}_0\}|$ instead of $s$ to account for the initialization $\mathcal{T}_0$ used here.
\hfill\qedb

\section{Proof of \Cref{thm:noise_algo_guarantee}}
\label{app:KF_based_alg}
We define $\vecd(k) = \vecx(k) - \vecx_f$. Our goal is to obtain an upper bound on $\bbE \{ \lVert \vecd (k) \rVert^2 \}$. First, from \eqref{eq:state_noise}, we have 
\begin{equation}
 \vecd(k) 
     \!=\! \matA \hat{\vecx}(k-1) \!+\! \matB \vecu(k-1) \!-\! \vecx_f \!+\! \matA\vecz(k-1) \!+\! \vecv(k-1),
\end{equation}
where $\hat{\vecx}(k)$ is the Kalman filter based estimate as defined in \Cref{line:kalman_est} of \Cref{alg:sparse_control_noise} and the state error $\vecz(k) \triangleq \vecx(k) - \hat{\vecx}(k)$. Since the noise $\vecv(k-1)$ has zero mean and is independent of all the other terms, taking the expectation of $\lV \vecd(k) \rV^2$ over $\{\vecv(i), \vecw(i)\}_{i=0}^{k-1}$, we derive
\begin{multline}\label{eq:step1}
 \mathbb{E} \lc \Vert \vecd(k) \Vert^2\rc = \mathbb{E} \lc\Vert \matA \hat{\vecx}(k-1) + \matB \vecu(k-1) - \vecx_f \Vert^2 \rc \\
 + 2\mathbb{E} \lc\lb \matA \hat{\vecx}(k-1) + \matB \vecu(k-1)\rb^{\T}  \matA \vecz(k-1)\rc \\- 
2\mathbb{E} \lc \vecx_f^{\T} \matA \vecz(k-1)\rc + \mathbb{E} \lc\Vert \matA \vecz(k-1) \Vert^2 \rc + \trace{\bs{\Sigma}_v}.
\end{multline}
We now analyze each of the terms on the right-hand side. We observe that, as the process and measurement noises are Gaussian, $\lc\vecx(i),\vecy(i)\rc_{i=1}^{k}$ is jointly Gaussian. Hence, the Kalman filter yields MMSE estimates. So, we deduce 
\begin{align} \label{eq:zero_bias}
    \bbE \lc\vecz(k-1)\rc = \zero, \quad \bbE \lc \vecz(k-1)|\lc\vecy(i)\rc_{i=1}^{k-1} \rc = \zero.
\end{align}
This observation allows us to deal with all the terms in \eqref{eq:step1} that are linear in $\vecz(k-1)$. Since $\vecx_f$ and $\matA$ are deterministic, from \eqref{eq:zero_bias}, we have
\begin{equation}\label{eq:step1_sub1}
   \mathbb{E} \lc \vecx_f^{\T} \matA \vecz(k-1)\rc=\vecx_f^{\T}\matA\mathbb{E} \lc \vecz(k-1)\rc=\zero.
\end{equation}
Further, we note that $\matA \hat{\vecx}(k-1) + \matB \vecu(k-1)$ is a function of the past measurements $\vecy(i)$, $i=1,2,\ldots,k-1$, and again from \eqref{eq:zero_bias}, the following term equals $\zero$,
\begin{multline*}
    \mathbb{E} \lc\lb\matA \hat{\vecx}(k-1) + \matB \vecu(k-1) \rb^{\T} \matA \vecz(k-1)\rc = \\
    \mathbb{E} \!\lc\!\lb\matA \hat{\vecx}(k-1) \!+\! \matB \vecu(k-1) \rb^{\T}\!\matA\mathbb{E}\! \lc \vecz(k-1)|\!\lc\vecy(i)\rc_{i=1}^{k-1}\!\rc\!\rc\!.
\end{multline*}
Combining \eqref{eq:step1_sub1} and the above, we simplify \eqref{eq:step1} as
\begin{multline}
 \mathbb{E} \lc \Vert \vecd(k) \Vert^2\rc = \mathbb{E} \lc\Vert \matA \hat{\vecx}(k-1) + \matB \vecu(k-1) - \vecx_f \Vert^2 \rc \\
 + \mathbb{E} \lc\Vert \matA \vecz(k-1) \Vert^2 \rc + \trace{\bs{\Sigma}_v}.\label{eq:step_main1}
\end{multline}
Next, we simplify the first term of the above equation using the exponentially decaying property of the residual norm of the OMP algorithm in \eqref{eq:omp_guarantee} as follows:
\begin{align} 
    &\mathbb{E} \lc\Vert \matA \hat{\vecx}(k-1) + \matB \vecu(k-1) - \vecx_f \Vert^2 \rc  \notag \\
    &\leq  \xi^s(\matB) \mathbb{E} \lc \Vert \matA \hat{\vecx}(k-1) - \vecx_f \Vert^2 \rc\\
    &= \xi^s(\matB)\mathbb{E} \lc \Vert \matA \vecx(k-1) \!-\! \vecx_f \Vert^2 \rc\! +\! \xi^s(\matB)\mathbb{E}\lc \Vert \matA \vecz(k\!-\!1) \Vert^2 \!\rc\notag \notag \\
    &\quad+ 2\xi^s(\matB)\mathbb{E}\lc \vecx(k-1)^{\T} \matA^{\T} \matA \vecz(k-1) \rc,\label{eq:step3}
\end{align}
which follows from \eqref{eq:step1_sub1}. Using $\vecx(k-1)=\vecd(k-1)+\vecx_f$,
\begin{multline}
    \mathbb{E}\! \lc \Vert \matA \vecx(k-1) - \vecx_f \Vert^2 \rc \!=\! \mathbb{E} \!\lc \Vert \matA \vecd(k-1) +\matA\vecx_f - \vecx_f \Vert^2\rc \\
    \leq 2\Vert \matA \Vert^2\mathbb{E} \lc \Vert\vecd(k-1)\Vert^2\rc +2\Vert\matA\vecx_f - \vecx_f \Vert^2,\label{eq:step3_sub1}
\end{multline}
due to the Cauchy-Schwarz inequality. Similarly, using $\vecx(k-1)= \hat{\vecx}(k-1)+\vecz(k-1)$, we rewrite the last term of \eqref{eq:step3} as
\begin{align}
\mathbb{E}\lc \vecx(k-1)^{\T} \matA^{\T} \matA \vecz(k-1) \rc 
=  \mathbb{E}\lc \Vert \matA \vecz(k-1) \Vert^2 \rc.\label{eq:step3_sub2}
\end{align}
where the last step is due to \eqref{eq:zero_bias}. Combining \eqref{eq:step_main1}, \eqref{eq:step3}, \eqref{eq:step3_sub1}, and \eqref{eq:step3_sub2}, and defining $\chi = 2 \xi^{s}(\matB) \Vert \matA \Vert^2$, we arrive at
\begin{align}
 \mathbb{E} \lc \Vert \vecd(k) \Vert^2\rc \notag\\
 &\hspace{-1.8cm}\leq \chi\mathbb{E} \lc \Vert\vecd(k-1)\Vert^2\rc
 +(1+ 3\xi^s(\matB))\mathbb{E}\lc \Vert \matA \vecz(k-1) \Vert^2\rc\notag\\
 &+ 2\xi^s(\matB)\Vert\matA\vecx_f - \vecx_f \Vert^2+\trace{\bs{\Sigma}_v},\\
 & \hspace{-1.8cm}= \chi^k  \Vert \vecd(0) \Vert^2+ \lb1 + 3\xi^s(\matB) \rb\sum_{i=0}^{k-1}  \chi^{k-i-1}  \mathbb{E} \Vert \matA \vecz(i) \Vert^2\notag\\
 & \hspace{-1.5cm} +\sum_{i=1}^k \chi^{i-1} \ls 2 \xi^s(\matB)\Vert \matA \vecx_f - \vecx_f \Vert^2+ \trace{\bs{\Sigma}_v}  \rs.
\label{eq:recus_sum_expansion}
\end{align}
We note that $\chi < 1$ since $s \geq \max\{ -\frac{\log(2\Vert \matA \Vert^2)}{\log(\xi(\matB))}, 0 \}$ by assumption. So, the first term of \eqref{eq:recus_sum_expansion} is a decreasing function of $k$, which is lower bounded by 0. Hence, we can always find an integer $k_1$ such that
\begin{equation} \label{eq:bound_1}
\chi^{k} \Vert \vecd(0) \Vert^2  \leq \frac{\eta}{3},
\end{equation}
for all $k \geq k_1$. 
Also, for the second term of \eqref{eq:recus_sum_expansion}, we note that 
\begin{equation}\label{eq:az_norm}
    \mathbb{E} \Vert \matA \vecz(i) \Vert^2 = \trace{\matA \matP_i \matA^{\T} } .
\end{equation}
Substituting \eqref{eq:bound_1} and \eqref{eq:az_norm} into \eqref{eq:recus_sum_expansion} gives
\begin{multline}\label{eq:bound_last}
    \hspace{-5mm}\mathbb{E} \lc \Vert \vecd(k) \Vert^2\rc <  \frac{\eta}{3}+ \lb1 + 3\xi^s(\matB) \rb\sum_{i=0}^{k-1}  \chi^{k-i-1}  \trace{\matA \matP_i \matA^{\T} }\\
 +\frac{ 2 \xi^s(\matB)\Vert \matA \vecx_f - \vecx_f \Vert^2+ \trace{\bs{\Sigma}_v} }{1-2 \xi^{s}(\matB) \Vert \matA \Vert^2},
\end{multline}
where we use the property that, for any integers $0\leq i_i\leq i_f$, 
\begin{equation}\label{eq:gp_prop}
    \sum_{i=i_i}^{i_f} \chi^{i} <\sum_{i=0}^{\infty} \chi^{i} = \frac{1}{1-\chi},
\end{equation}
because $0<\chi<1$. Now, to establish the desired bound in \eqref{eq:noise_algo_guarantee}, it remains analyze the second term in~\eqref{eq:bound_last}. To this end, we use \eqref{eq:KF_cnvg_1} and \eqref{eq:KF_cnvg_2}  to obtain $\lim_{i\to\infty}\matP_i=\matP$. The limit $\matP$ exists when the system $(\matA,\bs{\Sigma}_v^{1/2},\matC)$ is controllable and observable. Hence, we can find an integer $k_2$ satisfying 
 \begin{equation}\label{eq:bound_2}
     \trace{\matA \matP_i \matA^{\T} } < \trace{\matA \matP \matA^{\T} }+ \frac{\eta(1-\chi)}{3(1+\xi^s(\matB)) },
 \end{equation}
 for all  $i \geq k_{2}$. Therefore, for $k>k_2$, the second term of \eqref{eq:bound_last} simplifies as follows:
\begin{align}
\sum_{i=0}^{k-1} \chi^{k-i-1} \trace{\matA \matP_i \matA^{\T}}  \notag\\
&\hspace{-3.8cm}\leq\!\sum_{i=0}^{k_2-1}   \chi^{k-i-1}  \trace{\!\matA \matP_i \matA^{\T}\!} \!+ \max_{i\geq k_2} \trace{\!\matA \matP_i \matA^{\T}\!} \!\sum_{i = k_2}^{k-1} \!\chi^{k-i-1} \nonumber \\
&\hspace{-3.8cm}\leq \!\sum_{i=0}^{k_2-1}   \chi^{k-i-1}  \trace{\!\matA \matP_i \matA^{\T}\!} \!+\frac{\trace{\!\matA \matP \matA^{\T}\!}}{1- \chi} + \frac{\eta}{3(1+\xi^s(\matB)) },\label{eq:step4_sub3}
\end{align}
due to \eqref{eq:gp_prop} and \eqref{eq:bound_2}. Next, we use $\chi < 1$ to choose an integer $k_3>k_2$ such that
\begin{equation*}
    \chi^{k_3-k_2}\sum_{i=0}^{k_2-1}  \chi^{k_2-i-1}  
\trace{\matA \matP_i \matA^{\T}} < \frac{\eta }{3\lb1+3\xi^s(\matB)\rb}. 
\end{equation*}
Hence, for all $k>k_3$, we bound the first term of \eqref{eq:step4_sub3} as
\begin{multline}
\sum_{i=0}^{k_2-1}  \chi^{k-i-1}  \trace{\matA \matP_i \matA^{\T}}\leq \sum_{i=0}^{k_2-1}   \chi^{k_3-i-1}  \trace{\matA \matP_i \matA^{\T}}  \\
< \frac{\eta }{3\lb1+3\xi^s(\matB)\rb}.\label{eq:step4_sub4}
\end{multline}
Now \eqref{eq:step4_sub3} and \eqref{eq:step4_sub4} jointly lead to
\begin{multline}\label{eq:Az_norm_bound}
    \lb1 + 3\xi^s(\matB) \rb\sum_{i=0}^{k-1}  \chi^{k-i-1}  \trace{\matA \matP_i \matA^{\T}}\\
    <\frac{1 + 3\xi^s(\matB)}{1-  2\xi^s(\matB) \Vert \matA \Vert^2 }\trace{\matA \matP \matA^{\T}} +\frac{2\eta }{3}.
\end{multline}
Combining \eqref{eq:bound_last} and \eqref{eq:Az_norm_bound}, we arrive at \eqref{eq:noise_algo_guarantee}.  Hence, the proof is complete.
 \hfill \qedb





\end{appendices}

\section*{References}
\bibliographystyle{Supporting_Files/IEEEtran.bst}
\bibliography{Supporting_Files/IEEEabrv.bib,refs}    

\begin{IEEEbiographynophoto}{Krishna Praveen V. S. Kondapi}
    (GS'24) is pursuing Ph.D. degree at Indian Institute of Science (IISc), Bengaluru, India. His research interests are sparse signal processing and networked control systems.
\end{IEEEbiographynophoto}
\begin{IEEEbiographynophoto}{Chandrasekhar Sriram}
    received the M.Tech. degree from IISc. He is currently working at Texas Instruments India Pvt. Ltd. His research interests are in the areas of energy-efficient signal processing, networked control systems, and compressed sensing.
\end{IEEEbiographynophoto}
\begin{IEEEbiographynophoto}{Geethu Joseph}
    (M'20) received the B.Tech. degree in ECE from the National Institute of Technology Calicut, India in 2011. She received the M.Eng. degree in signal processing and the Ph.D. degree in ECE from IISc. in 2014, 2019, respectively. She was a Postdoctoral Fellow with the Department of Electrical Engineering and Computer Science, Syracuse University, Syracuse, NY, USA, from 2019 to 2021. She is currently a tenured Assistant Professor with the Signal Processing Systems Group, Delft University of Technology, Delft, The Netherlands. Her research interests include statistical signal processing, network control, and machine learning.
\end{IEEEbiographynophoto}
\begin{IEEEbiographynophoto} {Chandra R. Murhty}
    (F'23) is a Professor at the Department of Electrical Communication Engineering, IISc, India. His research interests are in the areas of energy harvesting communications, multiuser MIMO systems, and sparse signal recovery techniques applied to wireless communications.
\end{IEEEbiographynophoto}
\end{document}